%% file: main.tex
\documentclass[a4paper,11pt]{article}%

\input{preamble}

\title{\huge Selling Multiple Complements with Packaging Costs}

\author[1]{Simon Finster}
\affil[1]{CREST-ENSAE and Inria/FairPlay, \href{mailto:simon.finster@ensae.fr}{simon.finster@ensae.fr}}
        
\date{8 February 2025}

\renewcommand*{\thefootnote}{\arabic{footnote}}

\begin{document}

\maketitle

% Abstract.
\begin{abstract}
\normalsize
\onehalfspacing
We consider a package assignment problem with multiple units of indivisible items. The seller can specify preferences over partitions of their supply between buyers as packaging costs. We propose incremental costs together with a graph that defines cost interdependence to express these preferences. This facilitates the use of linear programming to characterize Walrasian equilibrium prices. Firstly, we show that equilibrium prices are uniform, anonymous, and linear in packages. Prices and marginal gains exhibit a nested structure, which we characterize in closed form for complete graphs. Secondly, we provide sufficient conditions for the existence of package-linear competitive prices using an ascending auction implementation. Our framework of partition preferences ensures fair and transparent dual pricing and admits preferences over the concentration of allocated bundles in the market.\\

\end{abstract}

\blfootnote{\emph{Acknowledgments}: I thank Nick Arnosti, Elizabeth Baldwin, Martin Bichler, Lorenzo Croissant, Peter Cramton, Péter Es{\H o}, Pär Holmberg, Zsuzsanna Jankó, Ian Jewitt, Bernhard Kasberger, Bettina Klaus, Paul Klemperer, Maciej Kotowski, Bing Liu, Edwin Lock, Simon Loertscher, Axel Ockenfels, Bary Pradelski, Alex Teytelboym, Bert Willems, Zaifu Yang, Yu Zhou, and seminar and conference audiences at the Simons Laufer Mathematical Sciences Institute (Berkeley), the University of Oxford, the University of Cologne, the University of Naples Federico II, the Research Institute of Industrial Economics Stockholm, the 2020 Econometric Society/Delhi Winter School, EARIE 2021, and various meetings of the Econometric Society for helpful comments and suggestions.
        
Parts of this material are based upon work supported by the National Science Foundation under Grant No.~DMS-1928930 and by the Alfred P. Sloan Foundation under grant G-2021-16778, while the author was in residence at the Simons Laufer Mathematical Sciences Institute (formerly MSRI) in Berkeley, California, during the Fall 2023 semester.}

\onehalfspacing

\section{Introduction}
\label{sec:introduction}

In many markets, buyers express preferences for bundles of indivisible items, and the seller cares about how items are bundled.
So far, despite attested real-world preferences, e.g.,~caps on the number of licenses in wireless spectrum auctions \citep{Cramton-2011,Kasberger-2023,Myers-2023}, the economic literature has not addressed competitive equilibrium with preferences over the partitioning of a supply of indivisible goods.
A key challenge is the representation of such preferences, as their domain, the space of supply partitions, is much larger than the space of bundles.
Moreover, with indivisible goods, the existence of Walrasian equilibria is guaranteed only under specific assumptions (see, e.g.,~\cite{Bikhchandani1997,Gul-1999, Sun2006,Milgrom2009}).

This article proposes a framework to express partition preferences. We show that Walrasian equilibria in our market are supported by uniform and anonymous prices that are linear in packages and reflect the nested structure of the seller's partition preferences. Furthermore, we provide conditions for the existence of Walrasian equilibria.
Our setting is a competitive market for multiple copies of multiple indivisible goods, including a set of buyers with values over bundles of items and a seller with preferences over partitions of their supply between buyers. We impose some structure on the seller's preferences with two objectives: we aim to obtain Walrasian equilibria that do not require personalized pricing while allowing the seller to express general preferences over the concentration of their allocated supply. To do so, we propose \emph{incremental cost functions} which specify the cost (savings) of combining two or more items for assignment to a single buyer, together with a graph structure defining cost interdependencies between bundles. Buyers have a rich set of preferences over bundles that can incorporate complements and substitutes.

In our market, the seller cares about whether any two or more items are allocated to separate buyers or, packaged, to a single buyer. Thus, we name the associated costs \emph{packaging costs}. 
Consider, for example, the reallocation of land plots to farmers via an auction. The government favors the allocation of two complementary land plots as a bundle to encourage defragmentation (cf.~\cite{Bryan-2024}). However, it wishes to allocate the two land plots separately if individual buyers have sufficiently high values. In our framework, it can do so by defining \emph{negative packaging costs}, offering the bundle more cheaply than the sum of costs of the items contained.\footnote{This can also be interpreted as an indirect subsidy. Selling biodiversity conservation contracts (e.g.,~\cite{Stoneham2003}), one may wish to subsidize bundles if contracts implemented on the same land are synergetic.} A small number of highly productive plots the government wishes to sell separately for fairness reasons, which can be expressed as \emph{positive packaging costs}. Because an assigned bundle may affect the cost of other bundles, e.g.,~if they contain items of the same variety, the seller's overall cost depends on the entire partitioning of their supply.
Applications for our framework are also found in procurement, the insurance industry, the transport sector, or wireless spectrum auctions.%
\footnote{In a procurement market with multiple suppliers, the \emph{buyer} has partition preferences: a single supplier may be preferred for machine maintenance and employee training. Decision support systems used in practice allow for different types of discounts and sophisticated bids \citep{Giunipero2009,Bichler2011}. \cite{Bichler2011} propose a bidding language that allows various types of discounts on bundles, but not in the context of Walrasian equilibrium. In the insurance industry, the risk between bundled products is interdependent. Representing such interdependencies while maintaining tractability is a key challenge in the partitioning problem. For transport sector auctions, see, e.g.,~\cite{Cantillon2005}, and for spectrum auctions, e.g.,~\cite{Cramton-2011}.}
In practice, sellers often wish for flexibility in shaping the market outcome; e.g.,~the Bank of England can express complex supply curves for allocating loans to commercial banks \citep{Klemperer2008,Klemperer2010,Klemperer2018}.

The market objective is to find a socially efficient bundling of the seller's supply and assignment of bundles to buyers. For buyers, we allow substitute trade-offs between any bundles, thus also allowing complements. If a buyer is assigned a collection of bundles, their value corresponds to its value-maximal matching to fictitious unit-demand agents, where each unit-demand agent is matched with at most one bundle.

For our first set of results, we construct the social welfare maximization problem embedding the graph structure that is part of the seller's partition preferences into an integer program. We characterize dual prices and show that they satisfy the desired properties of uniformity, anonymity, and linearity in packages.\footnote{Under package-linear pricing, the same price applies to identical packages, and the price of a collection of several packages equals the sum of prices of the packages contained in the collection. A package-linear pricing function is non-linear in items, i.e.,~the price of a package need not be the sum of prices of the contained items.}
We prove that this pricing function supports a competitive equilibrium and that the integrality of a solution to the linear relaxation of the social welfare maximization problem is equivalent to the existence of competitive equilibria. The classical results of \citet{Bikhchandani1997} and \citet{Bikhchandani2002} are not applicable because our seller's cost function depends on the entire partition sold. Thus, the buyers' values and the seller's costs must be considered separately, in contrast to, e.g.,~simply additive costs between items. To embed the graph structure into the social welfare maximization problem, we provide an algorithm that represents the graph's characteristic function, which uniquely maps supply partitions to their associated costs. For complete graphs, we derive a closed-form expression of the characteristic function.

In our second set of results, we provide sufficient conditions for the existence of competitive equilibria and establish a duality between revenue-maximizing and utility-maximizing sellers. We show that if the seller's partition preferences only involve (weakly) negative packaging costs, i.e.,~she prefers coarser partitions of her supply, and buyers have (weakly) superadditive values, i.e.~items are weak complements, a package-linear Walrasian equilibrium exists. This generalizes a result of \citet{Sun2014}, who establish equilibrium existence when buyers and the seller have superadditive values.
We also prove a duality relation between the objectives of a (value-based) revenue-maximizing and a (cost-based) utility-maximizing seller. Moreover, we establish the notion of \emph{set-cover submodularity} (\cref{def:set-cover-submodularity}), a weaker requirement than submodularity and show that the objective of a revenue-maximizing seller with superadditive values is equivalent to that of a utility-maximizing seller with set-cover submodular costs \emph{without} partition preferences.

Our framework of partition preferences implies new market design applications. The cost function graph we introduce facilitates modeling a wide range of cost interdependencies between related packages. For example, the bundle $ABC$ may be more expensive if several copies of $AB$ are sold as well, but it may be independent of the bundles of type $AC$. Furthermore, our incremental costs can encourage or discourage bundle allocation, but always provide flexibility to bundle in the opposite way, i.e.,~to sell items separately or together, if demand requires it. To illustrate this, consider the seller in a spectrum auction with concerns about an asymmetric distribution of licenses between buyers \citep{Ofcom-2017,GSMA-2021}. Such preferences would commonly be expressed through (hard) spectrum caps, i.e.,~each bidder can win at most a fixed number of licenses for each frequency band \citep{Cramton-2013,Kasberger-2023}. Although caps are typically set to mitigate market power in the downstream market, ``their drawback is that they may prohibit efficient aggregation of spectrum'' \citep{Cramton-2011}. The auctioneer can implement softer caps with our incremental cost structure, by penalizing bundle allocations with higher packaging costs. A buyer could then still obtain a bundled set of licenses, but the soft caps would require their bundle value to be high enough to outweigh the seller's preference for less market concentration. More generally, an auctioneer or regulator can steer market outcomes towards their preferred allocation using our cost structure, where hard caps are a special case corresponding to prohibitively high costs for certain bundles.

Our buyer preferences generalize assignment valuations \citep{Shapley-1962,Shapley-1971} and those admissible in the Product-Mix Auction with positive bids in that our buyers express substitute trade-offs between bundles instead of items (see \cite{Baldwin2024language} for more detail on the Product-Mix bidding languages). Combinatorial preferences also appear in the early package auction \emph{i}Bundle \citep{Parkes1999}, with a comprehensive account given, e.g.,~in \cite{Nisan2000} and \cite{Lehmann2006} (see also \cref{sec:agents-and-preferences} and the discussion in \cref{sec:discussion}).
Product-Mix auctions implement a Walrasian equilibrium assuming competitive behavior of participants. In the same way, our market can be implemented as a sealed-bid auction if agents act (approximately) as price takers and truthfully submit their preferences. The buyers' values and the seller's incremental cost functions and admissible graph structures are parsimonious in the vast space of partitions.
In markets for weak complements with negative packaging costs, we provide an implementation as an ascending auction. Our extended ascending auction strictly generalizes the ascending auction by \cite{Sun2014}, allowing a revenue-maximizing seller and an auctioneer with partition preferences.

Walrasian equilibrium in markets in which agents view some or all items as complementary have been studied, e.g.,~by \cite{Sun2006} and \cite{Teytelboym2014}, who establish equilibrium existence results, and by \cite{Sun-Yang-2009}, who develop a Walrasian tâtonnement process. \cite{Baldwin2019-a} introduce the concept of demand types and establish with their unimodularity theorem the existence of a Walrasian equilibrium for many classes of complements and substitutes preferences. Furthermore, \cite{Candogan-2015} show the existence of a linear pricing Walrasian equilibrium for the class of sign-consistent tree valuations\footnote{In this class, each two goods must be either substitutes or complements for all buyers.} and \cite{Candogan-2018} study pricing equilibria when buyers have graphical valuations. In all of those studies, partition preferences over the market supply are not considered. We note that partition preferences with negative packaging costs as well as buyers' superadditive values are orthogonal to the well-known classes of gross substitutes \citep{Kelso1982} and strong substitutes \citep{Milgrom2009} which admit a linear-pricing Walrasian equilibrium.

The remainder of this article is structured as follows. In \cref{sec:model}, we describe our market and the partition preference framework. In \cref{sec:walrasian-equilibrium}, we show the desired properties of the equilibrium pricing function and characterize Walrasian equilibria. In \cref{sec:equilibrium-existence}, we prove sufficient conditions for the existence of a Walrasian equilibrium and explore the dual relationship between revenue-maximizing and utility-maximizing sellers. \cref{sec:discussion} provides a brief discussion and \cref{sec:conclusion} concludes.

\section{The Competitive Market with Packaging Costs}\label{sec:model}

\subsection{Preliminaries}\label{sec:preliminaries}

There are $n$ indivisible, distinguishable varieties (or items) in the economy, identified with $j \in N:=\{1,\dots,n\}$, with a supply of $\Omega_j$ units for each variety $j\in N$.
\begin{definition}[Package]
    A \emph{package} (or bundle) is a subset of $N$, i.e.,~an element $\Sin$.
\end{definition}
Packages allow a single unit of each variety to be bundled together. This is without loss of generality, since with appropriate labeling packages containing identical varieties can be mapped to our model (we exemplify this in \cref{sec:identical-items}).
Multisets allow multiple copies of items by including the multiplicity of elements in their ground set $N$, where occurrences of the same element are indistinguishable.\footnote{See, e.g.,~\cite{Blizard1990} for a detailed treatment of multisets.} Denoting by $\Z_+$ the set of positive integers including zero, we define a \emph{multiset} as a mapping $m:N\rightarrow \Z_+$, and represent it as a vector of the multiplicities of its elements $\bm{m}:= (m(1),\dots,m(n))$.
We work mainly with \emph{package multisets} with ground set $2^N$, represented by vectors $\kb = (k_{S_1},\dots,k_{S_{2^n}}) \in \msets$ where $k_{S}$ denotes the multiplicity of the package $S$.
\begin{definition}[Feasible multisets]
    The universe of all \emph{feasible} package multisets is given by
    \begin{align*}
        \K := \left\{ \kb \in \msets: \sum_{\substack{S\in 2^N\negthickspace}} k_S \mathbbm{1}_{j\in S} \leq \Omega_j ~~\forall j \in N \right\}.
    \end{align*}
\end{definition}
A package multiset can also be seen as an anonymous partition in which the elements of the partition are not labeled. We refer to an anonymous partition simply as partition, and use package multiset and partition synonymously.
If a package multiset, or partition, $\kb$ can be identified to a set $S$ (i.e.,~$k_S=1$ for exactly one $S\in 2^N$ and $k_S'=0$ for all other $S'$), we will abuse the formal definition and write $\kb=S$ for the sake of~clarity.

The set of package multisets over $2^N$ can be endowed with basic operations and functions. Fixing package multisets $\kb$ and $\kb'$, the sum $\kb'' = \kb + \kb'$ is defined by $k''_S = k_S + k'_S$ for all $\Sin$, and scalar multiplication as $\alpha \kb = (\alpha k_S )_{\Sin}$ for any $\alpha\in \R$. The cardinality of $\kb$ is given by $|\kb| = \sum_{\Sin} k_S$. The unpacking operator $^\mstar$ unpacks a package multiset $\kb$ into the multiset of varieties contained in $\kb$ so that $\kb^\mstar = (m_j)_{j\in N}$ with $m_j := \sum_{\Sin} k_S \mathbbm{1}_{j\in S}$ for all $j \in N$.

To simplify notation, we often use implicit summation $f(X,Y) = \sum_{x\in X,y\in Y}f(x,y)$ for any finite sets $X$ and $Y$. We let $[X]:=\{1,\dots,X\}$.

\subsection{Agents and Preferences}
\label{sec:agents-and-preferences}

The economy consists of a seller (``she''), denoted $0$, and a set of $L$ buyers (``he'') denoted $l \in \buyers:=\{1,\dots,L\}$ and we let $\agents := \buyers \cup \{0\}$.
The preferences of each buyer are specified by a value function $V^l: \msets \rightarrow \Z_+$ with $V^l(\emptyset)= 0$.
If a buyer demands at most one package, the value function is given by $V^l: 2^N \rightarrow \Z_+$. Such a buyer is called a ``unit-demand agent'', where the ``unit'' refers to a package. We restrict $V^l$ as follows: each buyer's value function is the aggregate of values of a finite number $\Qbar^l$ of fictitious unit-demand agents, and $\Qbar:= \max_{l\in\buyers} \Qbar^l$. We sometimes denote by $(q,l)$ the $q$th unit-demand agent of buyer $l$.
\begin{definition}[Unit demand valuation]
    The \emph{unit-demand valuations} of the fictitious agents associated with buyer $l$ are defined as $v^l: 2^N\times \Z_+ \rightarrow \Z_+$, where $v^l(S,q)$ is the value of bundle $S$ for the $q$\textsuperscript{th} unit-demand agent associated with buyer $l$.
\end{definition}
The value of a buyer $V^l$ for a multiset $\kb$ is obtained by a value-maximizing matching of the contained bundles to his fictitious unit-demand agents,\footnote{Such preferences are also known as assignment valuations \citep{Shapley-1962,Shapley-1971}, and the value-maximizing matching is sometimes called a maximum-weight matching. In the terminology of \cite{Lehmann2006}, unit-demand valuations are combined by an ``inclusive-or''-operation (see also \cite{Nisan2000}). A generalized version of assignment valuations is the assignment messages in \cite{Milgrom-2009b}. Note that each of our unit-demand agents may be assigned a bundle of items.} where each fictitious unit-demand agent is assigned at most one bundle.
Formally, we aggregate unit-demand values as follows. Let $S_q$ be the bundle assigned to unit-demand agent $q$ and $e^{S_q} \in \{0,1\}^{2^n}$ be the indicator vector with value $1$ for bundle $S_q$ and value $0$ for all other bundles. We assume $v^l(S,q)=0$ for all $q > \Qbar^l$.
\begin{definition}[Unit-demand value aggregation]
\label{def:marg-value-aggr}
    \begin{align*}
        V^l(\mathbf{k}) := \max_{\substack{S_q \subseteq 2^N}} \sum_{q\in [\Qbar^l]} v^l(S_q,q) 
        \quad\text{ s.t. } \sum_{q \in [\Qbar^l]} e^{S_q} \leq \mathbf{k}.
    \end{align*}
\end{definition}

A \emph{pricing function} is a function $p:\msets\rightarrow \mathbb{R}$ with $p(0) = 0$. It is \emph{nonlinear} in varieties $j \in N$, i.e.,~for any package $\Sin$, we may have $p(S) \neq \sum_{j\in S} p(j)$. A pricing function $p:\msets\rightarrow \mathbb{R}$ is \emph{package-linear} if and only if, for all $\kb \in \msets$, $p(\kb) = \sum_{\Sin} k_S p(S)$.
Thus, a package-linear pricing function can be represented as a mapping $p:2^N\rightarrow \mathbb{R}$.

Each agent's utility is quasi-linear and given by $u^l(\kb,p) = V^l(\kb) - p(\kb)$ when they receive a package multiset $\kb$ at package-linear prices $p$.
Note that a buyer who receives a partition $\kb$ cannot unpack $\kb$ (see definition in \cref{sec:preliminaries}), i.e.,~the individual bundles of $\kb$ are fixed. \cref{ex:running-example} below illustrates the aggregation of unit-demand values.

\begin{example}[label=running-example]\label{ex:running-example}
    Consider the sale of two units of good $A$ and two units of good $B$, which may be sold separately or in packages. The set $\{AB\}$ is considered a package.
    Suppose that there is one buyer with two corresponding unit-demand buyers. The unit-demand value functions are given in \cref{tab:example-marg-values} and aggregated according to \cref{def:marg-value-aggr}. For legibility, here we write multisets in set notation instead of vectors. We have, e.g.,~$v^1(\{A\}) = 3$, $v^1(\{B\}) = 5$, $v^1(\{A,B\}) = \max(3+2,5+1) = 6$, $v^1(\{A,AB\}) = 3+9$, $ v^1(\{B,AB\}) = 5+9$, and $v^1(\{A,B,AB\}) = \max(3+2,5+1,3+9,5+9,1+9,2+9) = 14$. Note that with only two unit-demand buyers, the third and any further bundles contribute a value of zero.
    \begin{figure}
        \centering
        \begin{subfigure}[b]{0.49\textwidth}
            \centering
            \includegraphics[scale=0.4]{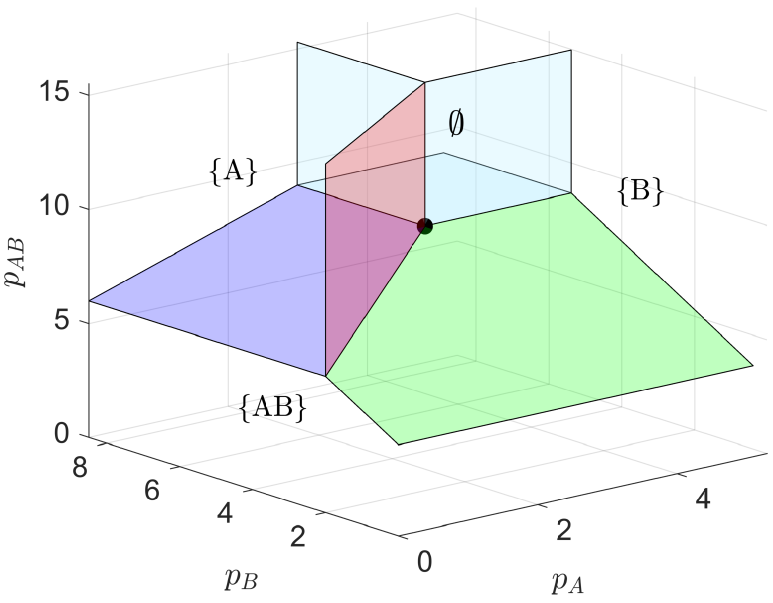}
            \caption{Unit demand in price space}
            \label{fig:3D-dot-bid}
        \end{subfigure}
        \begin{subfigure}[b]{0.49\textwidth}
            \centering
            \begin{tabular}{r|ccc}
                \toprule
                q     & $v^1(q,A)$ & $v^1(q,B)$ & $v^1(q,AB)$ \\
                \midrule
                1     & 3     & 5     & 9 \\
                2     & 1     & 2     & 9 \\
                \bottomrule
            \end{tabular}
            \caption{Unit-demand values}
            \label{tab:example-marg-values}
        \end{subfigure}
        \caption{Unit-demand values and their representation in price space}
        \label{fig:enter-label}
    \end{figure}
    The preferences of unit-demand agent $q=1$ are illustrated in price space in \cref{fig:3D-dot-bid} and represented by the black dot at $(3,5,9)$. At prices beyond the light-blue facets, the agent demands nothing. The remaining area is divided into three. At prices above the blue and left of the red facet $A$ is demanded, above the green and right of the red facet $B$ is demanded, and at prices below the blue and the green facet $AB$ is demanded. On the facets or their intersections, the agent is indifferent between the bundles demanded in adjacent regions.\footnote{The demand of multiple unit-demand agents can also be geometrically aggregated. \cite{Baldwin2019-a} illustrate this two-dimensional price space (not allowing a separate price for $AB$).}
\end{example}

The seller's preferences are over partitions of her supply, specified by the \emph{cost function} $C^0$. She has quasi-linear utility $u^0(\kb,p) =  p(\kb) - C^0(\kb)$ from selling the package multiset $\kb$ at a prices $p$. Any package multiset induces a partition of the contained bundles between buyers, where the identity of the buyers who receive a bundle is irrelevant. We note that, with multi-unit demand buyers (more than one corresponding unit-demand agent per buyer), there is an asymmetry in the interpretation of the seller's preferences: bundling items guarantees that a single buyer receives the bundle, but separate items may still be received by unit-demand agents belonging to the same buyer. If the seller cares about the separation of items, only one fictitious agent per buyer is allowed.

The cost function $C^0$ consists of two elements: 
(i) \emph{incremental costs} expressing additional costs or cost savings from bundling varieties in a package, and (ii) a graph with \emph{cost connections} specifying cost interdependencies between related bundles. For single-item packages $S$ with $|S|=1$, incremental costs are simply costs.
\begin{definition}\label{def:cfg}
    A \emph{cost function graph (CFG)} is a directed graph $G = (V,A)$ with vertices labeled with elements of $2^N$ representing the distinct packages, and arcs $A$ defining the cost connections between packages such that (i) if $(T,S) \in A$, then $S \subset T$, and (ii) every package $S$ is connected to the contained single varieties $\{j\}$, $j \in S$.
\end{definition}
The CFG is a tree due to property (i). Formally, we say that $S_1$ and $S_t$ are {(cost-)} connected and write $\exists (S_1...S_t)$ whenever there exists a sequence of vertices $(S_1,...,S_t)$ such that ${(S_1,S_2),...,(S_{t-1},S_t) \in A}$. Given a path $H:=(S_1,...,S_t)$, $|H|= t-1$ denotes the length of path $H$. Lastly, node $S$ is connected to itself.

The \textit{successors} $\Succes_0(S):=\{S'\in 2^N\colon\exists(S...S')\}$ and the \textit{strict successors} $\Succes(S):=\{S'\in 2^N\colon\exists(S...S'), S'\neq S\}$ of package $S$ are sets of packages reachable from $S$ in the CFG. Likewise, the \textit{predecessors} $\Predec_0(S):=\{S' \in 2^N\rt\exists(S'...S)$ and the \textit{strict predecessors} $\Predec(S):=\{S' \in 2^N\rt\exists(S'...S), {S'}\neq S\}$ of $S$ are sets of packages from which $S$ can be reached. Furthermore, $\DSucces(S):=\{S'\in 2^N \colon (S,S') \in A\}$ and $\DPredec(S):=\{S'\in 2^N \colon (S',S) \in A\}$ are the neighbors, the \textit{direct successors} and \textit{direct predecessors} of $S$, and $\DSucces_0:=\DSucces \cup \{S\}$ and $\DPredec_0:=\DPredec \cup \{S\}$.

When selling a supply partition, the seller incurs costs related to its individual elements, as well as their interaction. Each node (package) $S$ in the graph is associated with an incremental cost function. All predecessors $S'\in \Predec_0(S)$ of some node $S\in2^N$ are \textit{cost-connected} to $S$. Economically, this means that an incremental cost corresponding to $S$ also contributes to the cost of $S'$. Only bundles that are a subset or superset of each other can be cost-connected (\cref{def:cfg}~(i)), and for each bundle, the seller incurs at least the cost of the individual varieties contained in the bundle (\cref{def:cfg}~(ii)).
\begin{definition}\label{def:incr-marg-cost}
    Incremental costs are defined as $\DC: 2^N\times\N\rightarrow \Z$. The seller incurs the cost increase $\DC(S,r)$ from selling a copy of a package $T$ due to its cost connection to package $S$, when she sells $r-1$ copies of other packages cost-connected to $S$.
\end{definition}
Note that $\DC(\{j\},r) := \infty$ for $r>\Omega_j$, for all $j\in N$. Negative incremental costs, i.e.,~cost savings, are allowed. The total cost of a supply partition is obtained by adding all incremental costs associated with the bundles contained in the partition. For each bundle, the cost function graph defines the set of cost-connected bundles.
\begin{definition}\label{def:cost}
    The cost of selling a partition of supply (package multiset) $\kb \in \K$ is defined as $C^0:\K \rightarrow \Z_+$. Given a cost function graph and associated incremental costs, 
    \begin{align*}
        C^0(\kb) = \sum_{\substack{\Sin }} 
        \DC(S,[r_S])
        ,\quad \quad \text{with } r_S: = \sum_{\substack{S' \in \Predec_0(S)}} k_{S'}.
    \end{align*}
\end{definition}
Note that $r_S$ counts the copies of cost-connected packages (the predecessors of $S$) that are sold in the partition $\kb$. To illustrate the cost aggregation, suppose a single copy of bundle $S$ is sold and nothing else. Then the cost of $S$ is obtained by adding all incremental costs $\DC(S',1)$ of cost-connected bundles $S'$, i.e.,~the cost of $S$ is $\sum_{S'\in\Predec_0(S)}\DC(S',1)$. We often use the implicit summation $\DC(\Predec_0(S),r)$.
Although negative incremental costs are allowed, we assume that costs are non-negative, i.e.,~$C^0(\kb)\geq 0$ for all $\kb \in \K$, and we make the following monotonicity assumption on incremental costs.
\begin{assumption}[Increasing incremental cost]\label{ass:increasing-cost}
    For any package $\Sin$ and for all $r \geq 1$, it holds that $\DC(S,r) \leq \DC(S,r+1)$.
\end{assumption}
Intuitively, the more packages cost-connected to some package $S$ are sold, the more costly it becomes to sell an additional cost-connected package. This assumption is similar to one of increasing marginal costs, but note that, due to potential cost interdependencies, marginal costs can only be defined as a function of the entire partition sold. The following observations further illustrate the properties of CFGs.
\begin{observation}\label{observation:graph-properties}~
    \begin{itemize}
        \item [(a)] Node $N$ is a source and a CFG may contain other sources (\cref{def:cfg} (i)).
        \item [(b)] Node $\Sin$ is a sink iff $|S| = 1$ (\cref{def:cfg} (i) and (ii)).
        \item [(c)] A CFG is weakly connected (\cref{def:cfg} (ii)).
    \end{itemize}
\end{observation}
In \cref{ex:running-example-seller}, we show the simplest cost function graph and associated incremental cost functions with two distinct items $A$ and $B$. However, the expressive power of cost function graphs is better illustrated with at least three items $A$, $B$, and $C$, where the costs of packages with overlapping subsets of varieties may be connected. For example, the cost savings of grouping $ABC$ together may depend on the number of units of package $AB$ that are sold simultaneously in the market. We illustrate this in \cref{ex:three-goods}.

\begin{example}[continues=running-example]\label{ex:running-example-seller}
    The seller's preferences can be summarized by the incremental costs in \cref{tab:ex1:incremental-cost} and the CFG in \cref{fig:ex1:cost-relations}. Note that with only two goods, there exists only one valid CFG, and we illustrate more complex graphs in \cref{ex:three-goods}. For legibility, we write multisets not as vectors but in set notation.
    
    The total cost is given, as defined in \cref{def:cost}, as $C^0(\{A\}) = 1$, $C^0(\{B\}) = 1$, $C^0(\{AB\}) = 1 + 1 - 1$, $C^0(\{A,A\}) = C^0(\{B,B\} = 1 + 2$, $C^0(\{A,B\}) = 1 + 1$, $C^0(\{A,AB\}) = C^0(\{B,AB\}) = 1 + 2 + 1 - 1$, $C^0(\{A,B,AB\} = 1 + 2 +1 + 2 - 1$, and $C^0(\{AB,AB\}) = 1 + 2 +1 +2 -1 + 0$.
    \begin{figure}
    \centering
    \begin{subfigure}[b]{0.49\textwidth}
        \centering
        \includegraphics[scale=0.8]{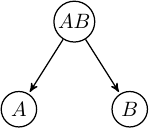}
        \caption{Cost function graph}
        \label{fig:ex1:cost-relations}
    \end{subfigure}
    \begin{subfigure}[b]{0.49\textwidth}
        \centering
        \begin{tabular}{r|ccc}
            \toprule
            r     & $\DC(A,r)$ & $\DC(B,r)$ & $\DC(AB,r)$ \\
            \midrule
            1     & 1     & 1     & -1 \\
            2     & 2     & 2     & 0 \\
            \bottomrule
        \end{tabular}
        \caption{Incremental cost}
        \label{tab:ex1:incremental-cost}
    \end{subfigure}
    \caption{Cost function graph for two goods $A$ and $B$ and incremental costs}
    \label{fig:ex1}
\end{figure}
\end{example}

\begin{example}\label{ex:three-goods}
    There are three goods $N=\{A,B,C\}$ that can be bundled as any package $S\in 2^N$. We consider two different cost function graphs, shown in \cref{fig:cfg_ABC} and \cref{fig:cfg_ABC_compl}.
    In \cref{fig:cfg_ABC}, each bundle $S$, $\vert S\vert~\geq 2$, when allocated to a buyer, creates a cost corresponding to its own incremental cost function and costs related to its \emph{subsets of a single variety}. With this type of packaging cost, the overall allocation of each variety affects the cost of related (superset) bundles. However, packaging costs are independent between bundles that consist of more than one item. This is illustrated in \cref{fig:mc_ABC} with the cost of selling the partition (or multiset) $\{A,B,C,AB,AC,BC, BC,ABC\}$.
    \begin{figure}[htb!]
    \centering
    \begin{subfigure}[b]{.4\textwidth}
      \centering
      \includegraphics[scale=0.8]{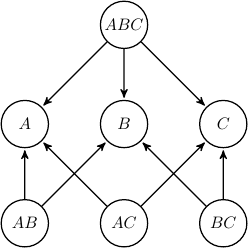}
      \caption{CFG}
        \label{fig:cfg_ABC}
    \end{subfigure}%
    \begin{subfigure}[b]{.59\textwidth}
      \centering
      \includegraphics[scale=0.3,trim=0 0 0 0, clip]{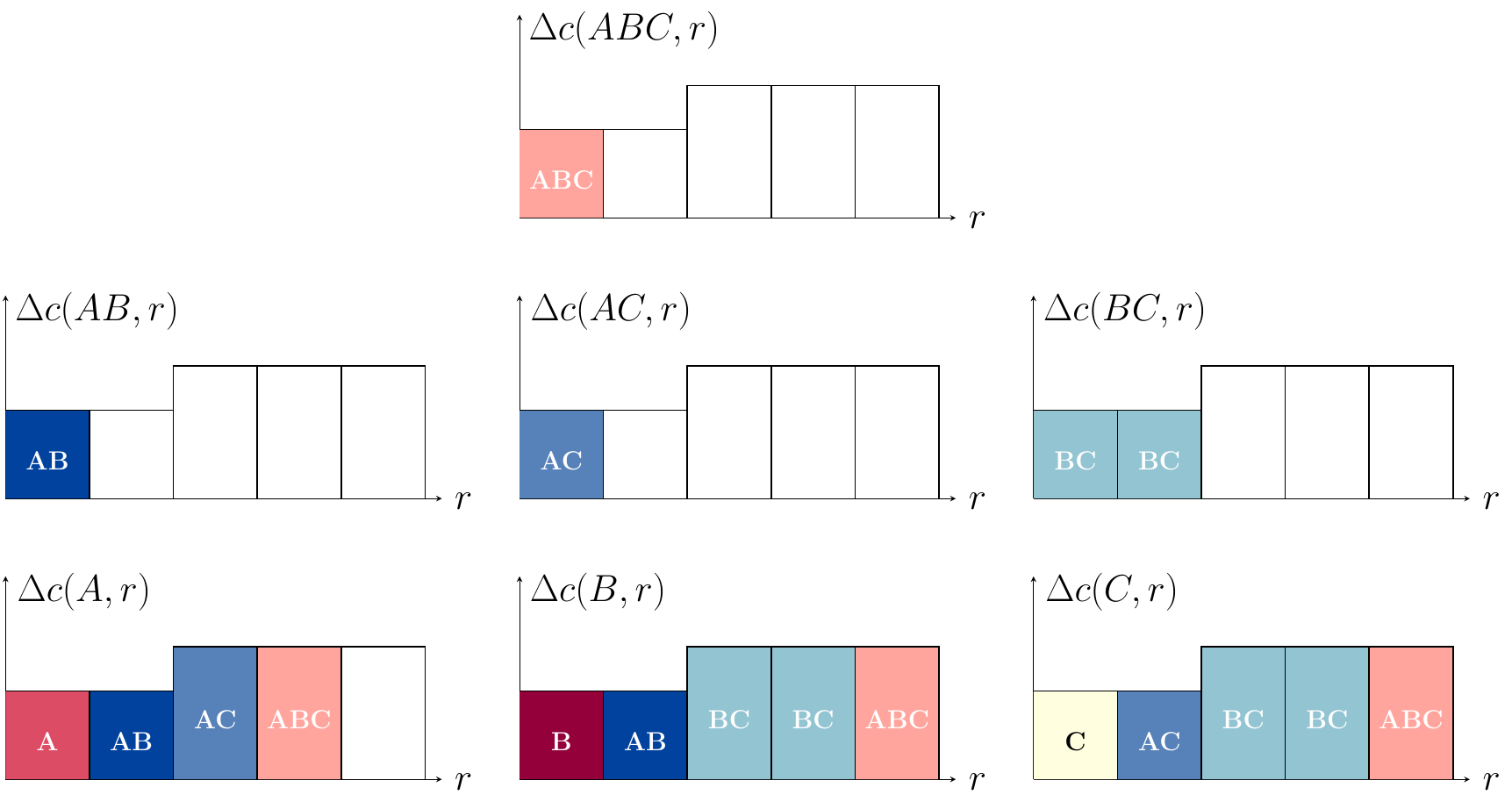}
        \caption{Incremental cost functions}
        \label{fig:mc_ABC}
    \end{subfigure}
    \caption{CFG and corresponding assignment on incremental cost functions of partition $\{A,B,C,AB,AC,$ $BC,BC,ABC\}$}
    \label{fig:CFG-and-cost}
    \end{figure}
    In \cref{fig:cfg_ABC_compl}, each bundle $S$, $\vert S\vert~\geq 2$, when allocated to a buyer, creates a cost corresponding to its own incremental cost function and costs related to \emph{all of its subsets}. Thus, the packaging cost of the sale of the package $ABC$ also depends on how many units of $AB$ (and $AC$ and $BC$) are being sold. Incremental cost functions are weakly increasing, so the more of, e.g.,~$AB$ is allocated, the more expensive it becomes to sell bundle $ABC$ (see \cref{fig:mc_ABC_compl}).
    \begin{figure}[tb!]
    \centering
    \begin{subfigure}[b]{.4\textwidth}
      \centering
      \includegraphics[scale=0.8]{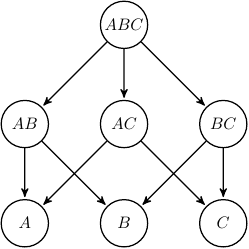}
      \caption{Complete CFG}
        \label{fig:cfg_ABC_compl}
    \end{subfigure}%
    \begin{subfigure}[b]{.59\textwidth}
      \centering
      \includegraphics[scale=0.3,trim=0 0 0 0, clip]{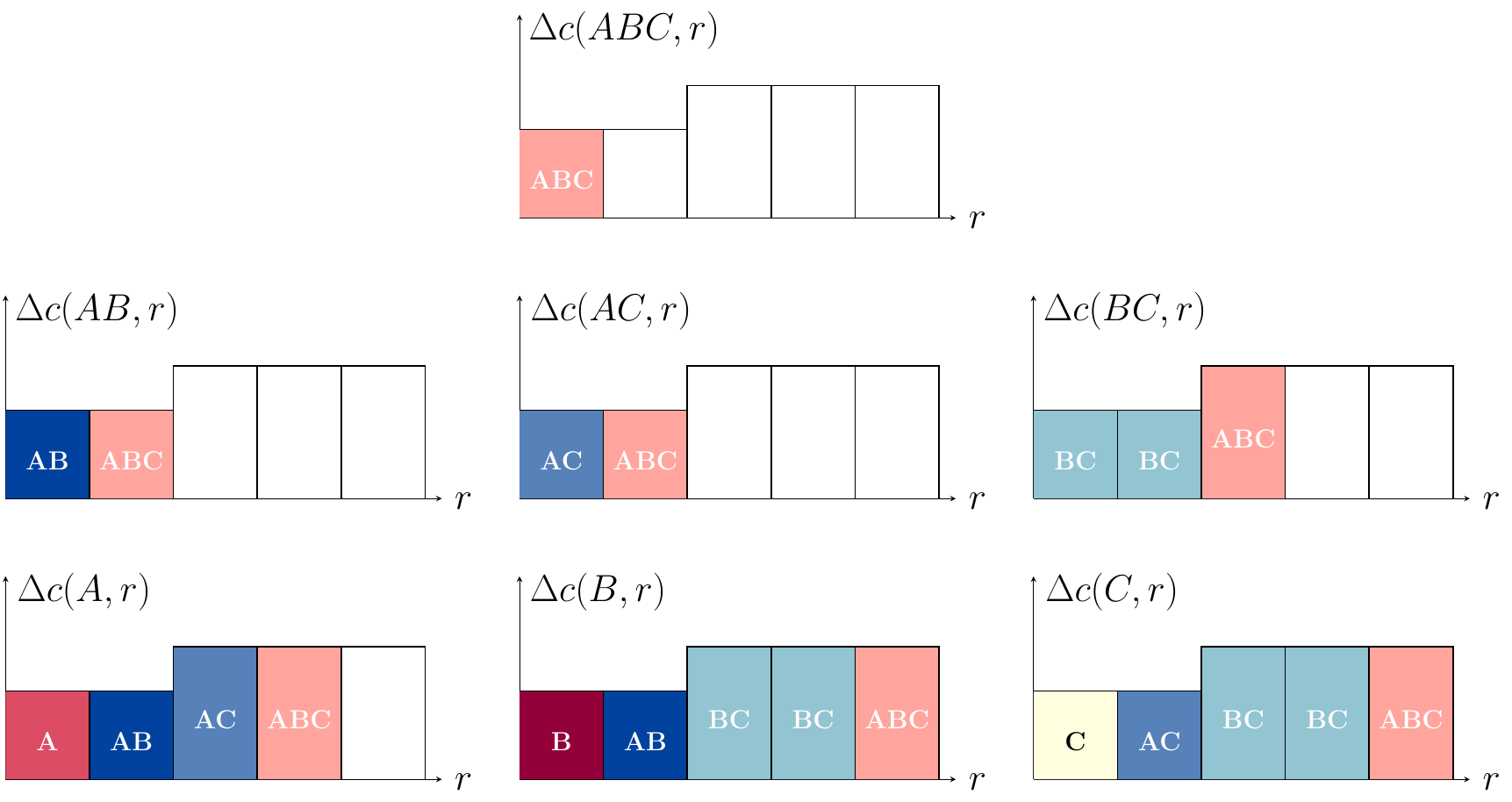}
        \caption{Incremental cost functions}
        \label{fig:mc_ABC_compl}
    \end{subfigure}
    \caption{Complete CFG and corresponding assignment on incremental cost functions of partition $\{A,B,C,AB,AC,$ $BC,BC,ABC\}$}
    \label{fig:CFG-and-cost-compl}
    \end{figure}
\end{example}
The two cost function graphs in the example above illustrate the most extreme cases of no cost interdependence and cost interdependence between all available bundles. More generally, the cost function graphs that satisfy \cref{def:cfg} allow a wide range of cost connections and associated packaging costs.

\subsection{Demand, Supply, and Equilibrium}
\label{sec:demand-supply-equ-multi-unit}

We define demand, supply, and Walrasian equilibrium. In our market, the equilibrium prices are \emph{package-linear}, i.e.,~the price of a partition is linear in the contained packages. The seller chooses how to partition and allocate their supply of individual varieties to buyers, i.e.,~she chooses a feasible multiset (anonymous partition) $\kb\in \K$ of packages to sell.

An allocation of items in $\Omega$ is defined as an assignment $\pi = (\pi(l))_{l \in \agents}$ of these items between the buyers and the seller, such that $\msetsum_{l \in \buyers} \pi(l) = \kb$ and $\pi(0) = \Omega - \kb^\mstar$. Recall that the operator $\mstar$ unpacks $\kb$ into a vector of individual varieties. $\pi(l)$ is the package multiset assigned to agent $l$ under the allocation $\pi$, where $\pi(l)$ may be the empty set, and $\pi(0) \neq \emptyset$ means that the items in $\pi(0)$ are not sold.
For all agents $l \in \agents$, the demand correspondences (or supply correspondence for the seller), and indirect utilities, are defined as $D^l(p):= \argmax_{\kb\in \K} u^l(\kb,p)$ and $ \mathcal{V}^l(p) := \max_{\kb\in \K} u^l(\kb,p)$.
An allocation $\pi$ is \emph{efficient} if, for every allocation $\pi'$, it holds that $\sum_{l\in\buyers}V^l(\pi(l)) - C^0(\kb) \geq \sum_{l\in\buyers}V^l(\pi'(l)) - C^0(\bm{k'})$, where $\kb = \msetsum_{l \in \buyers}\pi(l)$ and $\bm{k'} = \msetsum_{l \in \buyers} \pi'(l)$.
Given an efficient allocation $\pi$, the market value is defined as $V(\Omega):= \sum_{l\in\buyers}V^l(\pi(l)) - C^0(\kb)$. 

A \emph{package-linear pricing Walrasian equilibrium} is a tuple $(p^*,\pi^*)$, composed of a \emph{package-linear} pricing function $p^*(\kb) = \sum_{\Sin} k_S p^*(S)$, $\kb \in \msets$, with $p^*:2^N\rightarrow \mathbb{R}$ and an allocation $\pi^*$ such that $\msetsum_{l \in \buyers}\pi^*(l)\in D^0(p^*)$ and $\pi^*(l) \in D^l(p^*)$ for every buyer $l \in \buyers$.
Were the pricing function linear in varieties, the package-linear Walrasian equilibrium would reduce to the standard linear pricing Walrasian equilibrium.

\section{Walrasian Equilibrium}\label{sec:walrasian-equilibrium}
In this section, we formulate the social welfare maximization problem and characterize Walrasian equilibria via linear programming duality. The novelty in our results lies in that the seller's preferences are over supply partitions. Due to their nested structure of cost functions and graphs, we can summarize these preferences by the \emph{characteristic function} (\cref{lem:one-to-one-mapping}) and derive the pricing function. We then develop our main results starting with \cref{prop:pricing}, a characterization of competitive equilibrium prices, which exhibit a nested price structure related to the seller's partition preferences. Our \emph{characteristic function dual} relates prices and marginal gains from sale in \cref{prop:characteristic-dual}, and we confirm the important properties of uniformity, anonymity, and package-linearity in \cref{cor:uniform-anonmymous-package-linear-price}. In \cref{prop:primal-dual-CE}, we verify that this pricing function supports a competitive equilibrium allocation. Moreover, in \cref{prop:IP-LP-existence}, we show that the integrality of the LP solution is equivalent to the existence of a package-linear competitive equilibrium, generalizing the previous literature (e.g.,~\cite{Bikhchandani1997,Bikhchandani2002}) to incorporate partition preferences. Finally, we derive closed-form expressions for the characteristic function of complete graphs (describing the seller's partition preferences) in \cref{prop:charac-function-complete-cfg}. Our results on equilibrium existence and set-cover submodularity follow in \cref{sec:equilibrium-existence}.

As is standard in quasi-linear environments, the welfare theorems hold, stated in the following proposition. We provide a proof in \cref{proof:prop:efficientPLW-multi-unit} only for completeness.
\begin{proposition}\label{prop:efficientPLW-multi-unit}
     If $(p^*,\pi^*)$ is a package-linear Walrasian equilibrium, $\pi^*$ is an efficient allocation. If $\pi'$ is another efficient allocation, $(p^*,\pi')$ is also a package-linear Walrasian equilibrium.
\end{proposition}

\subsection{Social Welfare Maximization}\label{sec:welfare-maximization}

To formulate the social welfare maximization problem, we must aggregate the buyers' and the seller's preferences. As a first step, we reformulate the indirect utilities.
The buyers' value function can be rewritten using binary variables $x(S,q,l) \in \{0,1\}$, indicating whether fictitious unit-demand agent $q$ is assigned bundle $S$ or not.  We rewrite the indicator vectors in \cref{def:marg-value-aggr} as the constraint $\sum_{\Sin} x(S,q,l) \leqslant 1$, $\forall q \in\Qbarset$. Then, the buyers' indirect utility is given by
\begin{align}\label{equ:buyer-reformulation}
    \begin{split}
    V^l(\kb) & = \max_{\substack{\{x(S,q,l), \Sin, q \in \Qbarset\}}} \sum_{\substack{q \in \Qbarset, \Sin}} v^l(S,q) x(S,q,l) - \sum_{\Sin}k_S p(S) \\
    & \text{ s.t. } \quad x(S,\Qbarset,l) \leq k_S ~\forall S, \qquad x(2^N,q,l) \leq 1 ~\forall q, \qquad x(S,q,l) \in \{0, 1\} ~\forall S,q
    \end{split}
\end{align}
We denote by $\xb(S,l)$ the vector $(x(S,l,q))_{q\in\Qbarset}$, by $\xb(q,l)$ the vector $(x(S,l,q))_{\Sin}$, and by $\xb(l)$ the vector $(x(l))_{\Sin,q\in\Qbarset}$.

For the seller, we also introduce a binary variable $y(S,r)$. This variable indicates if partition $\kb$ invokes the assignment of step $r$ of the incremental cost function associated with bundle $S$. The cost function graph is encoded in the constraint in the seller's cost function in \cref{def:cost}, given by $r_s =\sum_{\substack{{S'} \in \Predec_0(S)}} k_{S'}$. Formally, $y(S,r)= 1$ if $r\leq r_s$ and zero otherwise. We assume that there exists a finite number of steps such that $\DC(S,r) < \infty$. We denote this number by $\Rbar := \max_{S,r}\{r:\DC(S,r) < \infty\}$ and $\Rbarset:=\{1,\dots,\Rbar\}$. 
\begin{lemma}\label{lemma:cost-function-assignment}
    If package $S$ appears in the seller's chosen partition, then one unit step on every incremental cost function corresponding to a successor of $S$, i.e.,~${S'}\in\Succes(S)$, must be assigned $y(S',\cdot) = 1$.
\end{lemma}
\begin{proof}
    This follows from the definitions of $r_S$ in \cref{def:cost} and of $y(S,r)$ above.
\end{proof}
\cref{lemma:cost-function-assignment} implies that the assignment on steps of an incremental cost function $\DC(S,\cdot)$ is limited by the minimum number of steps with finite height among all incremental cost functions $\DC(S',\cdot)$ of the successors $S'\in\Succes_0(S)$ of $S$. We write the seller's indirect utility as
\begin{align*}
    V^0(\kb) = \max_{\kb \in \K} \sum_{\Sin} k_S p(S) - \sum_{\substack{\Sin,r \in [r_s]}} \DC(S,r) y(S,r)\quad\text{ s.t. } r_s =\sum_{\substack{{S'} \in \Predec_0(S)}} k_{S'}, y(S,r)=1 ~\forall r\leq r_s
\end{align*}
Given a partition $\kb$, $y(S,r)$ for $\Sin$ and $r \in \N$ is uniquely defined. We also show that the converse holds, i.e.,~each CFG allocation $y(S,r)$ for $\Sin$ and $r \in \N$ can be mapped onto a unique partition of packages $\kb$. We formally state this in the lemma below and in Algorithm \ref{alg:one-to-one}. The proof is given in \cref{proof:lem:one-to-one-mapping}.
\begin{lemma}\label{lem:one-to-one-mapping}
    Given a cost function graph $G$, there exists a linear one-to-one mapping between an incremental cost function assignment $\{y(S,r)\}_{\Sin, r \leq \Rbar}$ and a corresponding package multiset $\kb$. This mapping, the ``characteristic function'' $\phi^G:\Z_+^{2^n\times \Rbar}\rightarrow \Z^{2^n}$, is inductively defined by Algorithm \ref{alg:one-to-one}.
\end{lemma}
Intuitively, Algorithm \ref{alg:one-to-one} starts with a source $\mathcal{S}$ in the graph, for which $k_{\mathcal{S}} = y(S,\Rbarset)$. Then, it selects some node $S$ for which all predecessors have already been visited. Because the predecessors have been visited, we can compute $k_S = y(S,\Rbarset) - \sum_{S'\in\Predec(S)}k_{S'}$. Then, it selects another node for which all predecessors have been visited, and so forth. For brevity, we write $\phi^G_S(\{S,\Rbarset\}_{\Sin})$ as $\phi^G_S(\{S,\Rbarset\})$. Crucially, the mapping is linear.\\

\begin{algorithm}[H]
    \SetAlgoLined
    \caption{Construct partition from CFG assignment}
    \label{alg:one-to-one}
    \SetKwInOut{Input}{Input}
    \Input{Cost function graph $G = (V,A)$ with assignment $(y(S,r))_{S\in V, r \leq \Rbar}$}
    Initialize list of successfully visited nodes $\mathscr{V}:=\emptyset$.\\
    \While{$V \neq \mathscr{V}$}
    { Select a node $S\in V\setminus \mathscr{V}$ such that $ V\setminus \mathscr{V}\cap \Predec(S) = \emptyset$ (all predecessors of $S$ have been visited) \\
    Set $k_S = y(S,\Rbarset) - \sum_{S'\Predec(S)}k_{S'}$ and $\mathscr{V} = \mathscr{V} \cup S$}
    \Return $(k_S)_{\Sin} = (\phi_S^G(\{y(S,\Rbarset)\}))_{\Sin} = \phi^G(\{y(S,\Rbarset)\})$
\end{algorithm}$~$

Using \cref{lem:one-to-one-mapping}, we rewrite the seller's problem as follows.
\begin{align}\label{equ:seller-reformulation}
    \begin{split}
    V^0(\kb) = & \max_{\{y(S,r), \Sin,r\in\Rbarset\}} \sum_{\Sin} k_S p(S) - \sum_{\substack{\Sin,r \in \Rbarset}} \DC(S,r) y(S,r) \\
    & \text{ s.t. } \quad k_S = \phi_S^G(\{y(S,r)\})~\forall S, \qquad y(S,r)\in \{0,1\} ~\forall S,r
    \end{split}
\end{align}

With the reformulation of the buyers' and the seller's utility maximization problem given in \cref{equ:buyer-reformulation,equ:seller-reformulation}, we can now state the social welfare maximization problem. The objective is to find the \emph{partition} $\kb \in \K$ of supply between buyers that maximizes the sum of the buyers' and the seller's utilities, and thus the buyers' values minus the seller's costs. Using the characteristic function $\phi^G$ from \cref{lem:one-to-one-mapping}, we can write the problem as an allocation problem with only binary decision variables and substitute $k_S$. The welfare maximization problem is named ``SWP''. For all constraints, we write short $\forall S$ for $\forall S\in 2^N$, $\forall q$ for $\forall q \in [\Qbar]$, $\forall r$ for $\forall r \in [\Rbar]$, and $\forall l$ for $\forall l \in \buyers$, and we use implicit summation.\\

{\bfseries SWP}
\begin{equation*}
    \max_{\{x(S,q,l),y(S,r),\Sin,q\in\Qbar,r\in\Rbar\}} \left[ \sum_{S,q,l} v^l(S,q) x(S,q,l) - \sum_{S,r}\DC(S,r)y(S,r)\right]
\end{equation*}
s.t.
\begin{flushright}
\noindent\begin{tabular}
    {%
    @{\ }>{$}r<{$}@{\ }%
    @{\ }>{$}c<{$}@{\ }%
    @{\ }>{$}l<{$}@{\ }%
    @{\ }>{$}l<{$}@{\ }%
    @{\ }>{$}l<{$}@{\hspace{85pt} }%
    @{\ }>{$}r<{$}@{}
    }%
    x(2^N,q,l) & \leq & 1 & \forall q,l &  & \constraint{equ:SWP1} \\  
    x(S,\Qbarset,\buyers) - \phi^G_S(\{y(S,r)\}) & \leq & 0 & \forall S & & \constraint{equ:SWP2} \\ 
    x(S,q,l), y(S,r) & \in & \{0,1\} & \forall S,q,r,l & & \constraint{equ:SWP3} \\ 
\end{tabular}
\end{flushright}
Prices in the objective function cancel because, in equilibrium, the demand and supply of bundles must be equal. Relaxing the integrality constraints, we write the integer program as the linear program ``SWLP'' with the corresponding dual variables listed next to the constraints. The feasible set of the SWLP is a non-empty, convex polytope, and therefore an optimal solution always exists; by strong duality, an optimal solution for its dual ``DSWLP'' also exists.\\

{\bfseries SWLP}
\begin{equation*}
    \max_{\{x(S,q,l),y(S,r),\Sin,q\in\Qbar,r\in\Rbar\}} \left[ \sum_{S,q,l}v^l(S,q)x(S,q,l) - \sum_{S,r}\DC(S,r)y(S,r)\right]
\end{equation*}
s.t.
\begin{flushright}
\noindent\begin{tabular}
    {%
    @{\ }>{$}r<{$}@{\ }%
    @{\ }>{$}c<{$}@{\ }%
    @{\ }>{$}l<{$}@{\ }%
    @{\ }>{$}l<{$}@{\ }%
    @{\ }>{$}l<{$}@{\hspace{80pt} }%
    @{\ }>{$}r<{$}@{}
    }%
    x(2^N,q,l) & \leq & 1 & \forall q,l & [b(q,l)] & \constraint{equ:SWLP1} \\ 
    x(S,\Qbarset,\buyers) - \phi^G_S(\{y(S,r)\}) & \leq & 0 & \forall S & [p(S)] & \constraint{equ:SWLP2} \\ 
    y(S,r) & \leq & 1 & \forall S,r & [d(S,r)] & \constraint{equ:SWLP3} \\ 
    x(S,q,l), y(S,r) & \geq & 0 & \forall S,q,r & & \constraint{equ:SWLP4}
\end{tabular}
\end{flushright}
\medskip

Note that the constraint $x(S,q,l)\leq 1$ is implied by the first constraint and can thus be omitted. The corresponding dual problem ``DSWLP'' is constructed using the characteristic function dual $\psi^G(\{p(S)\})$ of $\phi^G(\{y(S,r)\})$. Formally, let $\phi^G(\{y(S,r)\}) = \Phi \bm{y}^\intercal$, where $\Phi$ is a $2^n\times2^n$-matrix determined by Algorithm \ref{alg:one-to-one}, and $\bm{y} = \left( y(S,\Rbarset)\right)_{\Sin}$, i.e.,~a row vector each entry of which contains the total quantity allocated on incremental cost function $\DC(S,\cdot)$. Thus, we have $\phi^G_S(\{y(S,r)\}) = \Phi_S \bm{y}^\intercal$, i.e.,~the row corresponding to package $S$ of $\Phi$ multiplied by $\bm{y}^\intercal$. We define the characteristic function dual $\psi^G(\{p(S)\}) = {\Phi}^\intercal \bm{p}^\intercal$, where $\bm{p} = \left(p(S)\right)_\Sin$.\\

{\bfseries DSWLP}
\begin{equation*}
    \min_{\{b(q,l),p(S),d(S,r), \Sin, q\in\Qbarset,r\in\Rbarset,l\in\buyers\}}\left[
    \sum_{q,l} b(q,l) + \sum_{S,r} d(S,r) \right]
\end{equation*}%
s.t.
\begin{flushright}
\noindent\begin{tabular}
    {%
    @{\ }>{$}r<{$}@{\ }%
    @{\ }>{$}c<{$}@{\ }%
    @{\ }>{$}l<{$}@{\ }%
    @{\ }>{$}l<{$}@{\ }%
    @{\ }>{$}l<{$}@{\hspace{80pt} }%
    @{\ }>{$}r<{$}@{}
    }%
    b(q,l)+p(S) & \geq & v^l(S,q) & \forall S,q,l & [x(S,q,l)] & \constraint{equ:DSWLP1} \\ 
    d(S,r) - \psi^G_S(\{p(S)\}) & \geq & -\DC(S,r) & \forall S,r & [y(S,r)] & \constraint{equ:DSWLP2} \\
    b(q,l), d(S,r), p(S) & \geq & 0 & \forall S,q,r,l & &  \constraint{equ:DSWLP3}
\end{tabular}
\end{flushright}
\medskip

The pricing function is given as part of the DSWLP solution: $p(S)$ is the value of the last ``unit'' of a given bundle $S$ the seller assigns to some buyer. The dual variable $b(q,l)$ is the surplus of fictitious agent $q$ of buyer $l$, and $d(S,r)$ is the seller's incremental surplus from selling a bundle with a cost connection to $S$ assuming a corresponding incremental cost of $\DC(S,r)$.

We establish in \cref{prop:primal-dual-CE} that $p(S)$ is indeed a competitive equilibrium pricing function. This pricing function takes into account the interactions of a package $S$ with the aggregate partition $\kb$ the seller supplies. The price of the supplied partition, as well as the price for any buyer demanding a package multiset, is given by $p(\kb) = \sum_{\Sin} k_S p(S)$. We formally state the price structure in \cref{prop:pricing} below. We define $\rtild_S:= \argmax_r\{y(S,r):y(S,r)>0\}$, the last step on incremental cost function corresponding to bundle $S$, on which a positive quantity is allocated, and $\rtild:=\min_\Sin \rtild_S$.
\begin{lemma}\label{lemma:monotonicity-r-tilde}
    For all $\Sin$ and for all ${S'}\in\Succes(S)$, it holds that $\rtild_S \leq \rtild_{S'}$. 
\end{lemma}
\begin{proof}
    As noted in \cref{lemma:cost-function-assignment}, the assignment $y(S,r) = 1$ of a step on some incremental cost function $\DC(S,\cdot)$ requires the assignment of some step $r'$ such that $y({S'},r') = 1$ on every incremental cost function $\DC({S'},\cdot)$, if $S$ has a cost connection to ${S'}$, i.e.,~${S'}\in\Succes(S)$. Therefore, for any given partition, on any incremental cost function corresponding to ${S'}$, there must be at least as many steps with $y({S'},\cdot) = 1$ as on any incremental cost function corresponding to~$S$.
\end{proof}

\subsection{Competitive Equilibrium and its Pricing Function}\label{sec:competitive-equilibrium}
The linear programming formulation permits us to characterize the equilibrium pricing function, which exhibits a structure intimately related to the seller's partition preferences. The proof proceeds with an alternative, more complex formulation of SWLP. All proofs for this section are given in \cref{app:sec:proofs-walrasian-equilibrium}.
\begin{proposition}\label{prop:pricing}
    Fixing any $r \leq \Rbar$, it holds that $p(S) \leq \DC(\Succes_0(S),r) + d(\Succes_0(S),r)$ for all $\Sin$. For all $\Sin$ and $r\leq \rtild_S$, it holds that
    \begin{align*}
        p(S) = \DC (\Succes_0(S), r) + d(\Succes_0(S),r).
    \end{align*}
\end{proposition}
The above proposition characterizes prices as a function of the incremental costs associated with the bundles that are allocated to buyers, i.e.,~those assigned on cost function steps $r\leq \rtild_S$. The price of a bundle $S$ nests the incremental costs of all cost-connected bundles and the associated dual variables $d(S,r)$.

Furthermore, we characterize the characteristic function dual $\psi^G_S(\{p(S)\})$. Recall that the characteristic function $\phi^G$ maps the assignment on the incremental cost functions $\{y(S,r)\}_{\Sin, r \leq \Rbar}$ to the corresponding supply partition. Its dual $\psi_S^G$ maps the set of prices $\{p(S)\}_{\Sin}$ to the price gain from combining the individual items in $S$ together in package $S$, given all other cost connections. We formalize this in the following corollary.
\begin{corollary}\label{prop:characteristic-dual}
    For all $r \leq \rtild_S$, it holds that $\psi^G_S(\{p(S)\}) = p(S) - \left(\DC(\Succes(S),r) + d(\Succes(S),r)\right)$.
\end{corollary}
For the final unit of a bundle allocated, $\rtild_S$, we must have $\psi^G_S(\{p(S)\}) \leq \DC(S,\rtild_S+1)$, as for a price gain strictly greater than $\DC(S,r+1)$, the seller prefers to sell an additional bundle $S$.
The pricing function is a generalization of the uniform pricing rule for a market with multiple packages of different varieties and it can be shown to support an appropriate allocation as a competitive equilibrium.
\begin{proposition}\label{prop:primal-dual-CE}
    The pricing function $p^*(\kb) = \sum_{\Sin} k_S p^*(S)$, $\kb \in \msets$, where $\{p(S)\}_{S\in 2^N}$ is part of an optimal solution of DSWLP, supports the allocation $\{x^*(S,q,l)\}_{S\in 2^N, q \in [\Qbar^l], l \in \buyers}$, $\phi^G(\{y^*(S,r)\})$, given by an optimal solution of SWLP, as a package-linear pricing Walrasian equilibrium.
\end{proposition}
Because the seller has preferences over partitions, the proof requires additional techniques, compared to the standard LP literature, e.g.,~\cite{Bikhchandani1997}, and relies on \cref{prop:pricing}, \cref{prop:characteristic-dual}, and \cref{lem:facts} in the appendix.

The next corollary follows directly from \cref{prop:pricing} and \cref{prop:primal-dual-CE} and confirms the desired properties of equilibrium prices.
\begin{corollary}\label{cor:uniform-anonmymous-package-linear-price}
    The equilibrium pricing function $p(\kb)$ is uniform, anonymous, and package-linear.
\end{corollary}
Furthermore, we show that a Walrasian equilibrium exists if and only if it is characterized by an optimal solution of SWLP, further generalizing results of \cite{Bikhchandani1997} and \cite{Bikhchandani2002}.
\begin{proposition}\label{prop:IP-LP-existence}
    A package-linear pricing Walrasian equilibrium exists if and only if any optimal solution to SWP is also an optimal solution to SWLP, i.e.,~the optimum values of SWP and SWLP coincide.
\end{proposition}
To our knowledge, our \cref{prop:IP-LP-existence} is first to accommodate preferences over partitions of supply in the primal-dual characterization and existence equivalence of Walrasian equilibrium. Despite the generality of these partition preferences, their nested structure facilitates the pricing properties given in \cref{cor:uniform-anonmymous-package-linear-price}. Uniformity, anonymity, and linearity (in packages) are highly desirable in applications for reasons of fairness and transparency.

We also note that competitive equilibrium prices are not necessarily unique. This creates flexibility for the seller to choose from a set of equilibrium prices and to specify additional rules to do so.%
\footnote{Modern LP solvers can return the set of all integer solutions.} A set of lowest equilibrium prices for all packages may not exist as demonstrated in \cref{ex:running-example-3} below, but the seller may choose the lowest prices according to a lexicographic ordering.

In general, cost function graphs can be represented by their characteristic function, which can be computed using Algorithm \ref{alg:one-to-one}. For complete graphs, we also derive a closed-form solution of the characteristic function. First, we first define ``levels'' within a graph.
\begin{definition}\label{def:levels} 
    Let $x,y \in 2^N$. Then $x \subset_{t} y : = \{x \mid x \subseteq y, |y|-|x| = t\}$. $y$ is said to be $r$ levels above $x$, and $x$ is $r$ levels below $y$.
\end{definition}
Similarly, we define $x \subset_{\geq t} y$ and $x \subset_{\leq t} y$, whereby $|y|-|x| = t$ is replaced with $|y|-|x| \geq t$ and $|y|-|x| \leq t$, respectively. $x\supset_0 y$ implies $x=y$.
In a complete graph, each bundle has a cost connection to all of its subsets. This can be achieved with a minimal number of edges with each package $S$ pointing only to its subsets of size $|S|-1$. An example is shown in \cref{fig:cfg_ABC_compl}.
\begin{proposition}\label{prop:charac-function-complete-cfg}
    Let $G$ be a complete cost function graph. Then its characteristic function is given by $\phi^G_S(\{y(S,r)\}) = \sum_{t=0}^{n-|S|}\sum_{r,S' \supset_{t}S} (-1)^t y(S',r)$ for $\Sin$.
\end{proposition}
Using LP duality, we also derive a closed-form solution for the characteristic function dual.
\begin{corollary}\label{cor:charac-function-dual-complete-cfg}
    Let $G$ be a complete cost function graph. Then its characteristic function dual is given by $\psi^G_S(\{p(S)\}) = \sum_{t=0}^{|S|-1} \sum_{S' \subset_t S} (-1)^t p(S')$ for $\Sin$.
\end{corollary}
Complete cost function graphs may be especially relevant in the procurement of factor inputs. By way of example, suppose that if services $A$ and $B$ are delivered by the same provider, the buyer incurs cost savings, and similarly for services $A$ and $C$. However, if services $B$ and $C$ are delivered by the same provider, additional costs arise, e.g.,~because at least either $B$ and or $C$ are crucial to the production and the contingency risk of the provider. If all three services $A$, $B$, and $C$ are delivered by the same provider, the cost savings of $A$ and $B$, $A$ and $C$, and the additional costs of pairing $B$ and $C$ enter the cost function, plus an additional term to account for the interaction of $A$, $B$, and $C$.

\begin{example}[continues=running-example]\label{ex:running-example-3}
    \noindent Compared to the previous \cref{ex:running-example}, suppose there are two additional unit-demand agents labeled 3 and 4. All values are given in \cref{tab:ex1:incremental-cost-full}.
    The seller has cost savings if she sells one unit of $A$ and one unit of $B$ as a package; hence, agent 4 obtains $\{AB\}$ (breaking the tie between an alternative allocation of $A$ to agent 4 and $B$ to agent 1). For the second unit of $A$ and $B$, the seller is indifferent between selling items $A$ and $B$ separately and selling them as a bundle, so agent 1 and agent 3 win. The assignment of bundles sold on incremental cost functions is shown in \cref{fig:ex1:incr-cost-functions}. The cost of supplying bundle $AB$ consist of the cost of supplying $A$, $B$, and the packaging cost (savings) $\Delta c(AB,1)$.
    One can verify that the set of equilibrium prices is given by $(p(A),p(B),p({AB})$  $\in$ $\{(4,5,9)$, $(5,4,9)$, $(5,5,9)$, $(5,5,10)\}$, which all support the unique equilibrium allocation.
    From \cref{prop:pricing} we have that, e.g.,~$p({AB}) = p(A) + p(B) + \DC({AB},1) + d({AB},1)$, where $d({AB},1)$ can be either $0$ or $1$.

    \begin{figure}
        \centering
        \begin{subfigure}[b]{0.49\textwidth}
        \centering
        \includegraphics[scale=0.42]{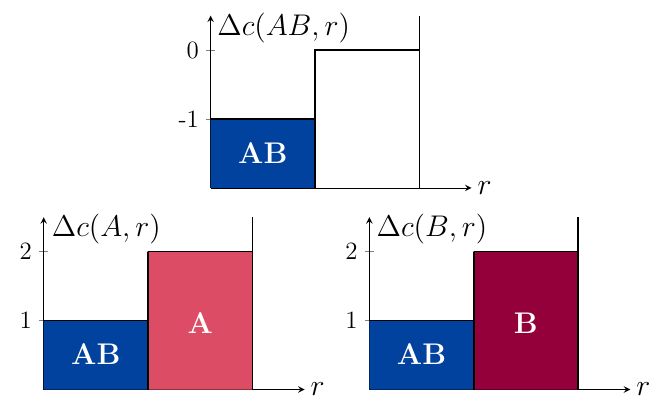}
        \caption{Incremental cost functions}
        \label{fig:ex1:incr-cost-functions}
        \end{subfigure}
        \begin{subfigure}[b]{0.49\textwidth}
        \centering
        \begin{tabular}{r|ccc}
        \toprule
        q     & $v^1(q,A)$ & $v^1(q,B)$ & $v^1(q,AB)$ \\
        \midrule
        1     & 3     & \underline{5}     & 9  \\
        2     & 1     & 2     & 9  \\
        3     & \underline{5}     & 3     & 8  \\
        4     & 6     & 2     & \underline{11} \\
        \bottomrule
        \end{tabular}
        \caption{Unit-demand values}
        \label{tab:ex1:incremental-cost-full}
        \end{subfigure}
        \caption{Bundles assigned on incremental cost functions and buyers' values with allocations}
        \label{fig:enter-label}
    \end{figure}
\end{example}

\section{Equilibrium Existence}\label{sec:equilibrium-existence}

We now establish sufficient conditions for the existence of package-linear Walrasian equilibria, and we show a duality between revenue-maximizing and utility-maximizing sellers (\cref{sec:revenue-vs-utility}). Market supply is restricted to one unit per variety, and we assume that the seller's partition preferences only comprise (weakly) negative packaging costs and buyers have (weakly) superadditive values, that is, they view items as weak complements. We prove that under these conditions, there exists a package-linear Walrasian equilibrium, using an ascending auction (\cref{sec:constructing-an-equilibrium}).

Given a market supply of $N$, supply partitions are described by vectors $\partition\in \{0,1\}^{2^n}$. Each buyer has a valuation $v^l: 2^N \rightarrow \Z_+$ and the seller's marginal cost of selling a package $\Sin$ is given by $c^0:2^N\rightarrow \Z_+$. These marginal costs can be made up of any arbitrary incremental cost function and cost function graphs as described in \cref{sec:agents-and-preferences}. The aggregate cost of selling a partition $\partition$ is given by $C^0(\partition) = \sum_{\Sin} c^0(S) k_S$.
We also write a partition now simply as $\partition=\{S_1,\dots,S_k\}$ where $S_t \in 2^N, t=1,\dots,k$, $S_{t_1}\cap S_{t_2} = \emptyset$ for all $S_{t_1}\neq S_{t_2}$, and $\bigcup_{t=1}^k S_t \in 2^N$. An allocation $\pi = (\pi(0),\pi(1),\dots,\pi(L))$ is also interpreted as a partition of supply, where $\pi(l)$ is removed if it is the empty set. Utilities, demand and supply correspondences, and indirect utilities remain unchanged.
\begin{assumption}
    The buyers' valuations are superadditive, i.e.,~for all disjoint $S_1,S_2 \in 2^N$ and for all buyers $l\in\buyers$, $v^l(S_1) + v^l(S_2) \leq v^l(S_1 \cup S_2)$.
\end{assumption}
Superadditivity is the most general concept of complementarity, which contains supermodularity and gross complements \citep{Samuelson1974,Sun2014}.
\begin{assumption}
    The seller's marginal costs are subadditive, i.e.,~for all disjoint $S_1,S_2 \in 2^N$, $c^0(S_1) + c^0(S_2) \geq c^0(S_1 \cup S_2)$.
\end{assumption}
We show the existence of package-linear Walrasian equilibrium by construction with a modified version of the ascending auction by \cite{Sun2014} (henceforth SY), and we extend this modified ascending auction to strictly generalize the auction by SY. We state the theorem and give the proof in \cref{sec:constructing-an-equilibrium}.
\begin{theorem}\label{th:equilibrium-existence}
    If the buyers' value functions satisfy superadditivity and the seller's marginal cost function satisfies subadditivity, there exists a package-linear Walrasian equilibrium.
\end{theorem}

\subsection{Constructing an Equilibrium}\label{sec:constructing-an-equilibrium}

In each round of the ascending auction, the seller states her supply, and buyers state their demand at current prices. If a package is overdemanded (defined below), the price of this package increases by one in the next round. The procedure stops as soon as no package is overdemanded.
We let $p(t,S)$ denote the price of bundle $\Sin$ at time $t$. The seller and the buyers behave straightforwardly. Buyer $l$ bids straightforwardly with respect to his values $v^l$ if, at all times $t \in \Z_+$ and for any prices $p(t)$, he demands $S^l(t) \in D^l(p(t)) = \argmax_{\Sin} \left\{v^l(S) - p(t,S) \right\}$, where $S^l(t) = \emptyset$ when $\emptyset \in D^l(p(t))$. That is, in each round, he demands a bundle that maximizes his utility given the current prices. We denote aggregate demand by partition $\partition^D(t):=\{S^1(t),\dots,S^L(t)\}$, which can also be seen as a vector $\partition^D\in \Z_+^{2^n}$ with each element counting how many buyers demand a package $\Sin$. The seller behaves straightforwardly with respect to her cost $C^0$ if, at any time $t\in \Z_+$ and at any prices $p(t)$, she chooses a supply partition $\partition(t)\in D^0(p(t)) = \argmax_{\partition\in\K} \sum_{S\in\partition}\left(p(t,S) - c^0(S)\right)$. We say that a package $S$ is overdemanded at time $t$ iff $k^D_S(t) > k_S(t)$, where $\partition^D(t)$ and $\partition(t)$ are the aggregate demand and the supply reported at time $t$, respectively. The procedure is described in Algorithm \ref{alg:ascending-auction}.

\begin{algorithm}[htbp]
    \SetAlgoLined
    \caption{Ascending auction}
    \label{alg:ascending-auction}
    Seller states initial reserve prices $c^0$ setting $t=0$ and initial prices $p(0,S) = c^0(S)~\forall\Sin$.\\
    All buyers $l\in \buyers$ report a demanded bundle $S^l(0)$ and the seller chooses a supply partition $\partition(0)$.
    \While{some package $S$ is overdemanded at time $t$}
    { Set $p(t+1,S) = p(t,S) + 1$ for all overdemanded bundles $S$, i.e.,~those with $k^D_S(t) > k_S(t)$.\\
    Set $p(t+1,S') = p(t,S')$ for all bundles $S'$ that are not overdemanded at $t$.\\
    Set $t = t +1$.
    }
    The auction terminates at $t = t^*$. Every bundle $S \in S^l(t^*)$ is allocated to buyer $l$ at price $p(t^*,S)$.\\
    \For{any $S\in\partition(t^*)$ not demanded by any buyer at $p(t^*)$}
    {
    \lIf{$p(t^*,S) = c^0(S)$}
    {$S$ remains with the seller}
    \Else{$S$ is allocated at price $p(t^*,S)$ to a buyer who demanded and was among the last to forfeit $S$ in a previous round.}
    }
    \Return Walrasian equilibrium prices $p(t^*,S)$, $\Sin$ and equilibrium allocation $\pi$
\end{algorithm}
\begin{proposition}\label{prop:ascending-auction}
    If all buyers bid straightforwardly, the ascending auction given by Algorithm \ref{alg:ascending-auction} terminates in a package-linear pricing Walrasian equilibrium after a finite number of rounds.
\end{proposition}
The proof is given in \cref{proof:prop:ascending-auction} and \cref{th:equilibrium-existence} follows then directly from \cref{prop:ascending-auction}. The ascending auction specified in Algorithm \ref{alg:ascending-auction} is modified from the ascending auction in SY by the participation of the seller with partition preferences. We also extend the procedure to fully generalize the ascending auction of SY in \cref{sec:extended-ascending-auction}. First, however, we demonstrate the difference and relation of a (value-based) revenue-maximizing seller, as it appears in SY, and our (cost-based) utility-maximizing seller.

\subsection{Revenue vs.~Utility Maximization}\label{sec:revenue-vs-utility}

The market studied by SY includes a revenue-maximizing seller whose preferences are symmetric to buyers. The seller's superadditive value function $v^0:2^N\rightarrow \Z_+$ is called a ``reserve price'' and she obtains the market price of all sold bundles plus the value of the unsold bundle as revenue.\footnote{With bundling, true reserve prices differ. Consider the supply of $A$ and $B$, which can be bundled as $AB$. The seller's values are $v(A)$, $v(B)$, and $v(AB)$, respectively. The seller chooses to sell $A$ alone only if $p(A) + v(B) > v(AB)$ and $p(A) + v(B) > p(AB)$. Thus, her true reserve price for $A$ is $v(AB) - v(B)$.}
In contrast, our seller has \emph{preferences over partitions} expressed through a cost function. The cost of supplying several bundles is the sum of their marginal costs, and the seller's utility is defined as the sum of obtained prices minus total cost.

In the following, we provide an equivalence characterization between value-based revenue maximization and cost-based utility maximization. This equivalence implies fundamentally different preferences in the model of SY and ours, and it reveals a new class of preferences that may be of independent interest.
We let $S^c$ denote the complement for any set $\Sin$, i.e.,~$S^{\mathrm{c}} = N\setminus S$ and define the set function dual (see, e.g.,~\cite{Gul2000} or \cite{Fushijige2005}). All proofs are deferred to \cref{proof:prop:seller-equivalence}.
\begin{definition}
    For any $\Sin$, given a set function $f:2^N\rightarrow \mathbb{R}$ with $f(\emptyset) = 0$, define the transformation $ g(f,S) = f(N) - f(S^{\mathrm{c}})$. $g(f,\cdot)$ is called the \emph{set function dual} of~$f$.
\end{definition}
Note that $g(f,N) = f(N)$. With the set function dual, we are equipped to formalize the relationship between revenue maximization and utility maximization.
\begin{proposition}\label{prop:seller-equivalence}
    Given a value function $v^0:2^N\rightarrow \Z_+$ with $v^0(\emptyset) = 0$, the objective of maximizing revenue is equivalent to maximizing utility, where the marginal costs $c^0:2^N\rightarrow \Z_+$ are given by the dual of $v^0$, and the total cost of any given partition $\kb$ is $C^0(\partition) = c^0\left(\bigcup_{S\in\partition}S\right)$.
\end{proposition}
We introduce a new notion of submodularity to further characterize the revenue-maximizing seller's cost function.
\begin{definition}[Set-cover submodularity]\label{def:set-cover-submodularity}
    Given a finite set $N$, a function $f:2^N\rightarrow \mathbb{R}$ is \emph{set-cover submodular} if $\forall~ S_1,S_2 \in 2^N$ with $S_1 \cup S_2 = N$,
    \begin{align*}
        f(S_1) + f(S_2)  \geq f(S_1\cup S_2) + f(S_1 \cap S_2).
    \end{align*}
\end{definition}
If the inequality sign is reversed, $f$ is \emph{set-cover supermodular}, and if replaced with an equality sign, $f$ is \emph{set-cover modular}.
The following proposition makes the crucial connection between superadditive and set-cover submodular functions.
\begin{proposition}\label{prop:setcover-submodular}
    Given a superadditive (subadditive) function $f:2^N\rightarrow \mathbb{R}$, its dual $g(f,\cdot)$ is set-cover submodular (set-cover supermodular).
\end{proposition}
From \cref{prop:setcover-submodular} and \cref{prop:seller-equivalence} we have the following corollary.
\begin{corollary}\label{cor:setcover-submodular-cost}
    The revenue-maximizing seller's marginal costs are set-cover submodular.
\end{corollary}
Set-cover submodularity is weaker than submodularity, because it is only required for every two subsets of $N$, the union of which fully covers $N$. In particular, set-cover submodularity does not imply subadditivity (unlike proper submodularity, which does imply subadditivity).
A set-cover submodular cost function can have strictly subadditive and strictly superadditive elements, as shown in \cref{ex:setcover-sub-super-additive} below. A revenue maximizer only takes into account the \emph{set} of sold items, whereas a utility maximizer with an appropriately defined cost function cares about the \emph{partition} of sold items between buyers.
\begin{example}\label{ex:setcover-sub-super-additive}
    The function $v$ in \cref{tab:superadd-and-submod} is superadditive and its dual $c$ is set-cover submodular. The dual $c$ is strictly subadditive with respect to $A$ and $BC$ $(2+3>4)$, but superadditive with respect to items $B$ and $C$ $(2+0<3)$.
    \begin{table}[htbp]
    \centering
        \centering
        \begin{tabular}{cccccccc}
            \toprule
               & $A$   & $B$   & $C$   & $AB$  & $AC$  & $BC$  & $ABC$ \\
            \midrule
            $v$ & 1     & 2     & 0     & 4     & 2     & 2     & 4 \\
            $c$ & 2     & 2     & 0     & 4     & 2     & 3     & 4 \\
            \bottomrule
        \end{tabular}
        \caption{Superadditive function $v$ and set-cover submodular dual $c$}
        \label{tab:superadd-and-submod}
    \end{table}
\end{example}
Although our seller's preferences are fundamentally different from a value-based revenue maximizer, we can generalize the ascending auction to incorporate agents with both types of value and cost functions.

\subsection{The Extended Ascending Auction}\label{sec:extended-ascending-auction}

In the extended auction, a revenue-maximizing seller as well as an additional seller with its own supply, or an additional agent with partition preferences over the original seller's supply, participate.
The revenue-maximizing seller surrenders her supply to the auctioneer and disguises herself as a buyer in the auction, attempting to buy back her supply. Her bids serve as reserve prices. In addition to potentially winning a bundle herself, she receives the price of every bundle (of her original supply) sold to proper buyers.
We consider a revenue-maximizing seller and an auctioneer with partition preferences over the seller's supply.\footnote{The variation of the market in which a second seller supplies an additional set of items that is disjoint from the other seller's supply is analogous, as long as items from the two sellers are not bundled together. If multiple sellers participated in the auction, a surplus sharing rule would have to be defined.}
\begin{lemma}\label{lem:same-allocation}
    If the auctioneer's marginal costs are zero, the extended ascending auction and the ascending auction by SY terminate in the same allocation (up to ties).
\end{lemma}
\begin{proposition}\label{prop:extended-ascending-auction}
    Suppose that the auctioneer's marginal costs are zero and the extended ascending auction starts at prices $p(-1,S) = v^0(S)-1 ~\forall \Sin$. Then, there exists a price path that is identical in the extended ascending auction and in the ascending auction by SY, resulting in final prices that support the same allocation in both auctions.
\end{proposition}
With non-zero marginal costs, the auctioneer can express additional preferences. Moreover, the existence of a package-linear Walrasian equilibrium is also guaranteed in the extended ascending auction. Thus, \cref{lem:same-allocation} and \cref{prop:extended-ascending-auction} imply that our auction strictly generalizes the ascending auction by SY. \cref{ex:ref_vs_utilmax} in the appendix further illustrates that the extended ascending auction can lead to different outcomes depending on the total cost function used, even when marginal costs and values are dual to each other.

In the literature, the existence of a Walrasian equilibrium has been demonstrated for markets with only substitutes (linear pricing) or only complements (non-linear pricing, cf.~SY), and some mixtures of substitutes and complements (cf.~\cite{Sun2006}), but not with preferences over partitions. Subadditive cost functions are especially relevant in markets with a high degree of fragmentation in which the auctioneer wishes to encourage bundle allocations. We also note that weak superadditivity of values allows buyers to be interested in a single item only or to have merely additive values.

\section{Discussion}
\label{sec:discussion}

\paragraph{Implementation as a sealed-bid auction.}
\label{sec:buyer-language}
Under the assumption of approximately competitive behavior, the Product-Mix Auction \cite{Klemperer2008,Klemperer2010,Klemperer2018} implements a competitive equilibrium with uniform, competitive prices in markets for substitutes. The Bank of England has been using this sealed-bid procedure to allocate loans to commercial banks against different types of collateral since the financial crisis in 2007. Similarly, our market may be implemented as a sealed-bid auction, assuming a large market with many buyers, in which all agents, including the seller, behave non-strategically. Our buyer preferences are related to those of buyers in the standard Product-Mix auction, where each buyer submits a finite list of ``paired bids'' bids to the auctioneer. This list of bids is equivalent to the aggregation of multiple unit-demand valuations (reported truthfully), and our buyer preferences generalize such Product-Mix preferences to incorporate packaged items.
We note that, although with a small number of strategic buyers efficiency is not guaranteed, if a large number of buyers for each package participate, our market is asymptotically ex-ante efficient \citep{Swinkels2001}.\footnote{Standard technical assumptions like asymptotic environmental similarity are required to be satisfied for each package separately. Related results have been shown by \cite{Cripps2006} and \cite{Fan2003}.}

Our seller's preferences also lend themselves as a bidding language, as the cost function graph and incremental cost functions provide a tractable structure that can be submitted to a clearing house or auctioneer. In general, in combinatorial allocation problems, there is no hope of achieving better than exponential communication, even if agents truthfully report their preferences \citep{Nisan-Segal-2006}. The aggregation of incremental costs significantly reduces the communication in the context bundles, as cost savings or increments for the combination of single varieties can be specified only for relevant bundles and omitted for those the seller is indifferent about.

\paragraph{The seller's demand type.}\label{sec:sellers-demand-type}

If all agents have strong substitutes valuations and supply is fixed, it is well known that a linear-pricing Walrasian equilibrium exists \citep{Milgrom2009}. The notion of strong substitutes can be extended to packages using the concept of demand types \citep{Baldwin2019-a}. An agent's demand type defines how demand may change due to small price changes, i.e.,~the possible trade-offs between items (in our setup, bundles). A valuation is strong substitutes if and only if it is concave and trade-offs are one-to-one.\footnote{\cite{Baldwin2019-a} show that a valuation is strong substitutes iff it is concave and corresponds to a strong substitutes demand type. A demand type is specified for quasi-linear utilities by a list of vectors describing the directions in which demand could change due to an infinitesimal generic price change. A strong substitutes demand type is defined by vectors with at most one +1 entry, at most one -1 entry, and no other non-zero entries. For details, we refer to \cite{Baldwin2019-a}.} Our buyers' preferences correspond to aggregations of unit-demand agents with one-to-one trade-offs and can thus be said to be \emph{``strong substitutes between packages''}. However, for this analogy between items and bundles to be complete, one would also require a fixed supply of bundles. In our market, supply is given in terms of items, and the seller divides these items into packages. Thus, in the context of package-linear pricing, her preferences are not strong substitutes between packages, and a package-linear pricing Walrasian equilibrium may not always exist.

We illustrate this with Example 2 from \cite{Sun2014} (also \cite{Bevia1999}).
Three buyers 1, 2, and 3 are interested in purchasing three items $A$, $B$, and $C$ with values given in \cref{tab:ex2-sun-yang}. The seller's costs are zero. Walrasian equilibrium does not exist, either with linear pricing \citep{Bevia1999} or with non-linear pricing \citep{Sun2014}. All buyers demand only one package, hence their valuations are strong substitutes between packages. However, induced by an infinitesimal price change, the seller may wish to, e.g.,~sell partition $\{AB\}$ instead of $\{A,B\}$. Thus, the characterization of the seller's ``demand'' type would contain the vectors $\pm(1,1,0,-1,0,0,0)$, where the entries correspond to the change in her supply of packages $(A,B,C,AB,AC,BC,ABC)$ induced by an arbitrarily small price change.
Therefore, the seller's demand type is not strong substitutes between packages.
\begin{table}[htp]
  \centering
    \begin{tabular}[t]{lrrrrrrrr}
        \toprule
        & $\emptyset$ & $A$ & $B$ & $C$ & $AB$ & $AC$ & $BC$ & $ABC$  \\
        \midrule
        Buyer 1 & 0     & 10    & 8     & 2     & 13    & 11    & 9     & 14     \\
        Buyer 2 & 0     & 8     & 5     & 10    & 13    & 14    & 13    & 15     \\
        Buyer 3 & 0     & 1     & 1     & 8     & 2     & 9     & 9     & 10\\
        \bottomrule
    \end{tabular}
    \caption{Buyers' valuations of bundles}
    \label{tab:ex2-sun-yang}
\end{table}

\section{Conclusion}
\label{sec:conclusion}

In markets for bundled items, it is often natural for the seller to have preferences over partitions of items between buyers. We analyze a competitive market in which such preferences are modeled as incremental cost functions together with a graph that defines cost interdependencies. We contribute the first characterization and existence results of anonymous and package-linear Walrasian equilibrium prices in the presence of partition preferences. We also uncover a duality relation between revenue- and utility-maximizing sellers, introducing the class of set-cover submodular valuations.
Although our results are theoretical, we also aim to inspire new market design applications in practice, especially in (near-)competitive environments. Our setting allows the seller to express vastly richer preferences than described in the previous literature and used in present-day auction design. The equilibrium prices we describe have a nested structure relating to the cost interdependencies specified by the seller, and ensuring transparency in pricing between related bundles. Moreover, in combinatorial auctions in practice, this may ease the construction of prices even for those bundles that were not bid for.
Finally, our framework of partition preferences may be of independent interest in other allocation problems, where the distribution of goods, services, or matchings in the market matters.

\renewcommand{\bibname}{References}

\bibliography{refs}

\pagebreak

\appendix

{\flushleft \Large \textbf{Appendix}}

\renewcommand{\thetable}{\Alph{section}\arabic{table}}
\renewcommand{\thefigure}{\Alph{section}\arabic{figure}}
\renewcommand{\theHfigure}{\Alph{section}\arabic{figure}}
\setcounter{figure}{0}

\section{Proofs}
\label{app:sec:proofs}
\setcounter{table}{0}
\setcounter{figure}{0}

\subsection{Proofs for \cref{sec:model}}
\label{app:sec:proofs-model}

\begin{proof2}[Proof of \cref{prop:efficientPLW-multi-unit}]\phantomsection\label{proof:prop:efficientPLW-multi-unit}
    Let $\mathcal{A}$ denote the universe of all feasible allocations, i.e.,
    \begin{align*}
        \mathcal{A}:= \left\{\pi \in \Z_+^{2^n\times|\agents|} : \msetsum_{l\in \buyers}\pi(l) = \kb \text{ for some } \kb\in \K, \text{ and } \pi(0) = \Omega - \kb^* \right\}
    \end{align*}
    $(p^*,\pi^*)$ is a package-linear pricing Walrasian equilibrium, so for any buyer $l\in\buyers$ and any allocation $\pi' \in \mathcal{A}$, we have
    \begin{align*}
        V^l(\pi^*(l)) - \sum_{S \in \pi^*(l)}  p^*(S) \geq V^l(\pi'(l)) - \sum_{S \in \pi'(l)}  p^*(S)
    \end{align*}
    Let $\msetsum_{l \in \buyers} \pi^*(l) = \kb$ and $\msetsum_{l \in \buyers}  \pi'(l) = \kb'$. We sum over $l\in\buyers$ and add and subtract the seller's cost
    \begin{align} \label{equ:proof-efficient-PLW}
    \begin{split}
        & \sum_{l\in\buyers} V^l(\pi^*(l)) - C^0(\kb) -  \left( \sum_{l\in\buyers} V^l(\pi'(l)) - C^0(\kb')\right)\\
        \geq & \sum_{l\in\buyers} \sum_{S \in \pi^*(l)} p^*(S) -C^0(\kb) - \left( \sum_{l\in\buyers}\sum_{S \in \pi'(l)} p^*(S) -  C^0(\kb') \right)
    \end{split}
    \end{align}
    Because $\pi^* \in D^0(p^*)$, we have, for all $\pi' \in \mathcal{A}$,
    \begin{align*}
        \sum_{l\in\buyers}  \sum_{S \in \pi^*(l)} p^*(S) -C^0(\kb)  \geq \sum_{l\in\buyers} \sum_{S \in \pi'(l)} p^*(S) -  C^0(\kb')
    \end{align*}
    From \cref{equ:proof-efficient-PLW}, it follows that, for all $\pi' \in \mathcal{A}$,
    \begin{align*}
             \sum_{l\in\buyers} V^l(\pi^*(l)) - C^0(\kb) -  \left( \sum_{l\in\buyers} V^l(\pi'(l)) - C^0(\kb')\right) \geq 0
    \end{align*}
    and so $\pi^*$ is efficient.
    
    Now let $\pi'$ be an efficient allocation. Then $V(\Omega)= \sum_{l\in\buyers}V^l(\pi'(l))- C^0(\kb')$. It also holds that $V(\Omega)= \sum_{l\in\buyers}V^l(\pi^*(l))- C^0(\kb)$ because $\pi^*$ is efficient as part of the equilibrium. The equilibrium is also buyer-optimal and seller-optimal, given prices. Thus, we obtain the following two inequalities:
    \begin{align*}
        & \mathcal{V}^l(p^*) \geq V^l(\pi'(l)) - \sum_{S \in \pi'(l)} p^*(S), \quad \quad \text{ for all } l \in \buyers \text{ and }\\
        & \sum_{l\in\buyers}\sum_{S \in \pi^*(l)}p^*(S) - C^0(\kb) = \mathcal{V}^0(p^*) \geq \sum_{l\in\buyers}\sum_{S \in \pi'(l)}p^*(S) - C^0(\kb')
    \end{align*}
    Suppose that one of these two inequalities were strict, then we would obtain
    \begin{align*}
        V(\Omega)
        & = \sum_{l\in\buyers}V^l(\pi^*(l)) - C^0(\kb)\\
        & = \sum_{l\in\buyers} \left[ V^l(\pi^*(l)) - \sum_{S \in \pi^*(l)} p^*(S) + \sum_{S \in \pi^*(l)} p^*(S) \right] - C^0(\kb)\\
        & = \sum_{l\in \buyers} \mathcal{V}^l(p^*) + \mathcal{V}^0(p^*) \\
        & > V^l(\pi'(l)) - \sum_{S \in \pi'(l)} p^*(S) + \sum_{l\in\buyers}\sum_{S \in \pi'(l)}p^*(S) - C^0(\kb')\\
        & = \sum_{l\in\buyers} V^l(\pi'(l)) - C^0(\kb') \\
        & = V(\Omega).
    \end{align*}
    This is a contradiction, and consequently it holds that 
    \begin{align*}
        \mathcal{V}^l(p^*) = V^l(\pi'(l)) - \sum_{S \in \pi'(l)} p^*(S), \quad \quad \text{ for all } l \in \buyers \text{ and },\\
        \mathcal{V}^0(p^*) = \sum_{l\in\buyers}\sum_{S \in \pi'(l)}p^*(S) - C^0(\kb'), \quad \text{ i.e., } \pi' \in D^0(p^*).
    \end{align*}
    It follows that $(p^*,\pi')$ is also a package-linear pricing Walrasian equilibrium.
\end{proof2}

\subsection{Proofs for \cref{sec:walrasian-equilibrium}}
\label{app:sec:proofs-walrasian-equilibrium}

\begin{proof2}[Proof of \cref{lem:one-to-one-mapping}]\phantomsection\label{proof:lem:one-to-one-mapping}
First, we show that given a cost function graph $G$, any incremental cost function assignment $(y(S,r))_{\Sin, r \leq \Rbar}$ maps into a unique vector $(k_S)_{\Sin}$. To proof that Algorithm \ref{alg:one-to-one} constructs a unique image from any permissible input, we demonstrate that (a) Algorithm \ref{alg:one-to-one} terminates after a finite number of steps, (b) each node $S$ is successfully visited at some point, and (c) the time at which $S$ is successfully visited is irrelevant. The mapping is linear by the definition of $k_S$ in the algorithm.
    
(a) and (b) follow because \cref{def:cfg}(i) implies that there are no cycles in $G$. Thus, as long as $V \neq \mathscr{V}$, there exists some $S\in V\setminus \mathscr{V}$ for which the if-condition is false. Consequently, all nodes are added to $\mathscr{V}$ at some point. It is without loss of generality to require that the algorithm does not get stuck in a trivial loop, i.e.,~it does not select a sequence of nodes that allow no successful visit and are skipped and revisited indefinitely. (c) follows because once the ``if condition'' in Algorithm \ref{alg:one-to-one} is false for a given node $S$, the set of nodes $A\in \mathscr{V}:\exists(A...S)$ remains unchanged. Once $S$ \emph{could} be successfully visited, it does not matter when it is actually selected for the successful visit, i.e.,~other nodes may be selected first.
    
The reverse mapping is straightforward. Given a cost function graph $G$, a partition $\kb = (k_S)_{\Sin}$ is mapped to the following incremental cost function assignment: for all $\Sin$, $y(S,r_S) = 1 ~\forall r_S \leq \sum_{A:\exists(A...S)} k_A$ and $y(S,r_S) = 0 ~\forall r > \sum_{A:\exists(A...S)} k_A$.
\end{proof2}

\hypertarget{proof:prop:pricing}{~}\\
\begin{proof2}[Proof of \cref{prop:pricing}]
    In the social welfare maximization problem, prices are nested in the dual characteristic function. To avoid this difficulty, we will use a more complex but equivalent formulation of the SWP. Recall that whenever a step on the incremental cost function $\DC(S,\cdot)$ is allocated, a step on every incremental cost function $\DC({S'},\cdot)$ for all $S'\in\Succes(S)$ must be also allocated (see also \cref{lemma:monotonicity-r-tilde}). Thus, on each incremental cost function associated with package $S$, some steps can be assigned that are directly linked to the allocation $S$ and, in addition, some steps may be allocated that are related to the allocation of some package $S'\in\Predec(S)$. To make this distinction, we denote a cost function step as $r({S'},S)$ if it is allocated on cost function $\DC({S'},\cdot)$ due to the allocation of package $S$ (to a buyer). With this definition, $\sum_{r(S,S)} y(S,r(S,S)) = k_S$.
    
    In the original SWLP, the cost function graph was encoded in the characteristic function $\phi^G$. In the subsequent alternative formulation SWLP$^\prime$, the cost function graph is encoded in the equality constraints which require that, for each incremental cost step corresponding to the allocation of package $S$, cost steps on all incremental cost functions corresponding to successors $\Succes(S)$ are also allocated. Recall that $S\in\Succes_0(S)$. In the notation of dual variables we write, e.g.,~$d(S',S,r(S',S))$ as $d(S',S,r)$ for brevity.
    \medskip
    
    {\bfseries SWLP$^\prime$}
    \begin{equation*}
        \max_{\substack{\{x(S,q,l),y({S'},r({S'},S)),r({S'},S)\in\Rbarset \\ \forall \Sin,{S'}\in\Succes_0(S)), q\in\Qbarset\}}}\left[ \sum_{S,q,l}v^l(S,q)x(S,q,l) - \sum_{S,{S'}\in\Succes_0(S),r({S'},S)} y({S'},r({S'},S)\DC({S'},r({S'},S))\right]
    \end{equation*}
    s.t.
    \begin{flushright}
    \noindent\begin{tabular}
    {%
    @{\ }>{$}r<{$}@{\ }%
    @{\ }>{$}c<{$}@{\ }%
    @{\ }>{$}l<{$}@{\ }%
    @{\ }>{$}l<{$}@{\ }%
    @{\ }>{$}l<{$}@{\ }%
    @{\ }>{$}r<{$}@{\ }
    }%
    x(2^N,q,l) & \leq & 1 & \forall q,l & b(q,l) & \constraint{equ:SWLP'1}\\
    x(S,\Qbarset,\buyers) - \sum_{r(S,S)}y(S,r(S,S)) & \leq & 0 & \forall S & [p(S)]& \constraint{equ:SWLP'2}\\
    y({S'},r({S'},S)) & \leq & 1 & \forall S, S'\in\Succes_0(S), r({S'},S) & [d({S'},S,r)] & \constraint{equ:SWLP'3}\\
    y(S,r(S,S)) - y({S'},r({S'},S)) & = &  0 & \forall S,S'\in\Succes(S), r({S'},S) & [u({S'},S,r)] & \constraint{equ:SWLP'4}\\ 
    x(S,q,l), y(S,r(S',S)) & \geq & 0 & \forall S,S'\in\Succes_0(S),q,r(S',S) & & \constraint{equ:SWLP'5}\\
    \end{tabular}
    \end{flushright}
    \medskip

    We formulate the dual of SWLP' as follows.\\
    \medskip
    
    {\bfseries DSWLP$^\prime$}
    \begin{equation*}
        \min_{\{b(q,l),p(S),d({S'},S,r),u({S'},S,r\}}\left[
        \sum_{q, l} b(q, l)+\sum_{S, {S'}\in\Succes_0(S), r({S'}, S)} d({S'}, S, r) \right]
    \end{equation*}%
    s.t.
    \begin{flushright}
    \noindent\begin{tabular}
    {%
    @{}>{$}r<{$}@{\ }%
    @{}>{$}c<{$}@{\ }%
    @{}>{$}l<{$}@{\ }%
    @{}>{$}l<{$}@{\ }%
    @{}>{$}l<{$}@{\ }%
    @{}>{$}r<{$}@{\ }
    }%
    b(q,l)+p(S) & \geq & v^l(S,q) & \forall S,q,l & [x(S,q,l)] & \constraint{equ:DSWLP'1}\\ 
    -p(S)+d(S,S,r) + u(\Succes(S),S,r) & \geq & -\DC(S,r(S,S)) & \forall S, r(S,S) & [y(S,r(S,S))] & \constraint{equ:DSWLP'2}\\
    d({S'},S,r)-u({S'},S,r) & \geq & -\DC({S'}, r({S'},S)) & \forall S,S'\in\Succes(S), r(S',S) & [y({S'},r({S'},S))] & \constraint{equ:DSWLP'3}\\
    b(q,l),d({S'},S,r),p(S) & \geq & 0, u({S'},S,r) \in \mathbb{R} & \forall S,S'\in\Succes(S), q, r(S',S) & & \constraint{equ:DSWLP'4}\\
    \end{tabular}
    \end{flushright}
    \medskip
    
    Substituting \cref{equ:DSWLP'3} into \cref{equ:DSWLP'2} and rearranging, we obtain $p(S) \leq \sum_{S'\in\Succes_0(S)} d({S'},S,r) + \DC({S'},r({S'},S))$.
    On any cost function step that is allocated a positive quantity $y({S'},r({S'},S))>0$ complementary slackness implies that constraints (\ref{equ:DSWLP'2}) and (\ref{equ:DSWLP'3})  hold with equality. While $r({S'},S)$ designates which package is allocated on step $r$ of incremental cost function $\DC({S'},\cdot)$, the order in which packages are allocated on incremental cost functions does not matter; it is only the sum of all allocated incremental cost steps that determines the seller's costs. Because $r({S'},S)$ could indeed be any of the steps on $\DC({S'},\cdot)$, on which a positive amount is allocated, we can omit the specification of $S$ and ${S'}$. The first statement of \cref{prop:pricing} holds then indeed for any step $r$, and the second statement for all steps $r$ on which a positive quantity is allocated, i.e.,~all $r\leq\rtild_S$.
\end{proof2}

\begin{observation}\label{observation:CSC}
    Complementary slackness from LP duality implies the following observations, corresponding to the constraints in \cref{equ:SWLP1,equ:SWLP2,equ:SWLP3,equ:DSWLP1,equ:DSWLP2}.
    \begin{flushright}
    \begin{tabular}{%
        @{\ }l@{\ }%
        @{\ }l@{\hspace{45pt} }%
        @{\ }r@{\ }%
        }%
        (i) & If $x(2^N,q,l) < 1$, then $b(q,l) = 0~\forall q\in\Qbarset,l\in\buyers$.
        & \csclabel{csc:S1}{S1}\\
        (ii) & If $x(S,\Qbarset,\buyers)  - \phi^G_S(\{y(S,r)\}) < 0$, then $p(S) = 0~\forall\Sin$.
        & \csclabel{csc:S2}{S2}\\
        (iii) & If $y(S,r) < 1$, then $d(S,r) = 0~\forall \Sin,r\leq\Rbar$.
        & \csclabel{csc:S3}{S3}\\
        (iv) & If $x(S,q,l) \neq 0$, then $b(q,l) = v^l(S,q) - p(S) ~\forall\Sin, q\in\Qbarset,l\in\buyers$.
        & \csclabel{csc:D1}{D1}\\
        (iv) & If $y(S,r) \neq 0$, then $d(S,r) + \DC(S,r) = \psi^G_S(\{p(S)\})~\forall\Sin,r\leq\Rbar$.
        & \csclabel{csc:D2}{D2}\\ 
    \end{tabular}
    \end{flushright}
\end{observation}

\begin{lemma}\label{lem:facts}
    If $\{x(S,q,l),y(S,r)\}$ and $\{b(q,l),d(S,r),p(S)\}$, $\Sin, q\in\Qbarset,l\in\buyers,r\in\Rbarset$, are solutions to SWLP and DSWLP, respectively, the following facts hold.
    \begin{flushright}
    \begin{tabular}{%
        @{\ }l@{\ }%
        @{\ }l@{\ }%
        @{\ }r@{\ }%
        }%
        (i) & $\sum_{S'\in\Predec_0(S)} x({S'}, \Qbarset, \buyers) \leq y(S,\Rbarset)$, for all $\Sin$. & \csclabel{fact:F1}{F1}\\
        (ii) & If $p(S) > 0$, then $\sum_{S'\in\Predec_0(S)} x({S'},\Qbarset,\buyers) = y(S,\Rbarset) $, for all $\Sin$. & \csclabel{fact:F2}{F2}\\
        (iii) & If $v^l(S,q) - p(S) < \max_{S'\neq S}\{v^l(S',q)-p(S'),0\}$, then $x(S,q,l)=0 $, for all $S,q,l$. & \csclabel{fact:F3}{F3}\\
        (iv) & If $\max_{\Sin}\{v^l(S,q) - p(S)\} > 0$, then $x(2^N,q,l) = 1 $, for all $q\in\Qbarset,l\in\buyers$. & \csclabel{fact:F4}{F4}\\
        (v) & If $\DC(S,r) < \psi^G_S(\{p(S)\})$, then $y(S,r) = 1 $, for all $ \Sin$, $r\leq \Rbar$. & \csclabel{fact:F5}{F5}\\
        (vi) & If $\DC(S,r) > \psi^G_S(\{p(S)\})$, then $y(S,r) = 0 $, for all $ \Sin$, $r\leq \Rbar$. & \csclabel{fact:F6}{F6}\\
    \end{tabular}
    \end{flushright}
\end{lemma}
\medskip

\begin{proof2}[Proof of \cref{lem:facts}]\phantomsection\label{proof:lem:facts}
\hyperref[fact:F1]{(F1)} is obtained by summing up the bundle $S$ supply constraints from SWLP. In particular, we sum over all $S'\in\Predec_0(S)$. By definition of incremental cost functions, $\sum_{S'\in\Predec_0(S)} Y_{S'} = \sum_{r\in\Rbarset}y(S,r)$ for any package $\Sin$.

To show \hyperref[fact:F2]{(F2)}, first note the contrapositive of \hyperref[csc:S2]{(S2)}: if $p(S) > 0$, $x(S,\Qbarset,\buyers) - \phi^G(\{y(S,r)\}) = 0$ ($>0$ is ruled out by the constraints \cref{equ:SWLP2} of the SWLP). We wish to sum these constraints across all packages $S'\in\Predec_0(S)$ of which some positive quantity is allocated, i.e.,~$y({S'},r)>0$ for some step $r$. For $y({S'},r)>0$, we can apply \cref{prop:pricing}, and with $C^0(\kb)\geq 0 ~\forall \kb\in\K$ and $d(S,r)\geq 0~\forall S,r$ it follows that $p({S'}) > 0 $ for all $S'\in\Predec_0(S)$. Therefore, we can take the sum of tight supply constraints corresponding to packages $S'\in\Predec_0(S)$, noting that including those packages ${S'}$ of which nothing is allocated, i.e.,~$y({S'},r)=0$ for all $r$ (and thus also $x({S'},q,l) = 0$ for all $q,l$), does not change the sum. Thus, \hyperref[fact:F2]{(F2)} follows.

\hyperref[fact:F3]{(F3)} is derived from \hyperref[csc:D1]{(D1)}: Assume $x(S,q,l)>0$ and $x(S'',q,l)>0$, $S\neq S''$. Then, by \hyperref[csc:D1]{(D1)}, $v^l(S,q)-p(S) = v^l(S'',q)-p(S'') = b(q,l)$. So if $v^l(S,q)-p(S) < v^l(S',q)-p(S')$ for some $S'$, then $x(S,q,l) = 0$. Furthermore, $b(q,l) \geq 0$ and $b(q,l) \geq v^l(S,q)- p(S)$ by \cref{equ:DSWLP3}. Thus, if $v^l(S,q)-p(S) < 0$, then $v^l(S,q)-p(S) < b(q,l)$, and thus $x(S,q,l)=0$. Together, we obtain \hyperref[fact:F3]{(F3)}.

To show \hyperref[fact:F4]{(F4)}, note that $b(q,l)\geq v^l(S,q) - p(S)$ implies that, if $\max_\Sin\{v^l(S,q) - p(S)\} > 0$, then $b(q,l)>0$. The contrapositive of \hyperref[csc:S1]{(S1)} then implies $x(2^N,q,l) = 1$, and thus \hyperref[fact:F4]{(F4)}.
\hyperref[fact:F5]{(F5)} follows by the contrapositive of \hyperref[csc:S3]{(S3)}, as from the constraint $d(S,r) \geq \psi^G(\{p(S)\}) - \DC(S,r)$ it follows that if $\DC(S,r) < \psi^G(\{p(S)\})$, then $d(S,r) >0$ for any $\Sin$ and $r \in \Rbarset$.
Finally, if $\DC(S,r) > \psi^G(\{p(S)\})$, then it must be $d(S,r) > \psi^G(\{p(S)\}) - \DC(S,r)$, because $d(S,r)$ is non-negative. Hence, the contrapositive of \hyperref[csc:D2]{(D2)} implies $y(S,r) = 0$, and therefore \hyperref[fact:F6]{(F6)}.
\end{proof2}\\

\begin{proof2}[Proof of \cref{prop:characteristic-dual}]\phantomsection\label{proof:prop:characteristic-dual}
    From \cref{prop:pricing}, we have $p(S) = \DC (\Succes_0(S), r) + d(\Succes_0(S),r)$. Substituting \hyperref[csc:D2]{(D2)} into \cref{prop:pricing}, the first statement follows.
    The contrapositive of \hyperref[fact:F5]{(F5)} gives $y(S,\rtild_S+1) = 0 \Rightarrow \DC(S,\rtild_S+1) \geq \psi_S^G(\{p(S)\})$. As $y(S,\rtild_S) = 1$, we can substitute \hyperref[csc:D2]{(D2)} and the second statement follows.
\end{proof2}\\

\begin{proof2}[Proof of \cref{prop:primal-dual-CE}]\phantomsection\label{proof:prop:primal-dual-CE}
Assume the allocation $\{x(S,q,l),y(S,r)\}$ and prices $\{p(S)\}$ are solutions of SWLP and DSWLP as defined above. By \cref{lem:facts}, conditions \hyperref[fact:F1]{(F1)} - \hyperref[fact:F6]{(F6)} hold. In the following, we show that, together with \hyperref[csc:S1]{(S1)} - \hyperref[csc:D2]{(D2)} and the constraints of SWLP and DSWLP, \cref{lem:facts} implies that the prices $\{p(S)\}$ support $\{x(S,q,l),y(S,r)\}$ as a package-linear pricing Walrasian equilibrium.

We prove that (a) there is no surplus improvement possible for any fictitious agents $q$ and any buyer $l$, (b) for a buyer who received a package multiset, no surplus improvement is possible from reassigning elements of the multiset to different unit-demand agents, (c) that, given an allocated supply partition and prices, the seller cannot improve her utility by allocating more or less of a given package, and (d) that, given an allocated supply partition and prices, the seller cannot improve her utility by choosing a different partition.

(a): If the surplus of unit-demand agent $(q,l)$ is negative on all packages in its valuation $v^l(S,q)$, \hyperref[fact:F3]{(F3)} ensures that this unit-demand agent is not assigned anything. \hyperref[fact:F4]{(F4)} implies that, if a strictly positive surplus can be made on any bundle of some unit-demand agent, the maximum quantity of one is allocated to that agent, and from \hyperref[fact:F3]{(F3)} it follows that only bundles that maximize the surplus of agent $(q,l)$ are allocated with a positive quantity. \cref{equ:SWLP1} ensures that not more than the maximum quantity of one is allocated.

(b): If a buyer receives a multiset of items $\kb$, it is value-maximally assigned to his corresponding unit-demand agents by \cref{def:marg-value-aggr}. Recall that the unpacking of multisets is not allowed by our model assumptions. To see that \emph{at the given auction prices}, a buyer has no incentive to reassign elements of the allocated multiset between his unit-demand agents, let $\kb$ denote the multiset of items received by buyer $l$ and let $\mathcal{Q}^l$ denote the set of the corresponding winning unit-demand agents, i.e.,~$\kb = \left(\sum_{q\in\mathcal{Q}^l} x(S,q,l)\right)_{\Sin}$. The buyer's utility is given by $u^l(\kb,p) = \sum_{q\in\mathcal{Q}^l,S} (v^l(S,q) - p(S)) x(S,q,l)$. Suppose that in an alternative assignment $\wtild \xb(l)\neq \xb(l)$ such that $\left(\sum_{q\in\mathcal{\wtild Q}^l} \wtild x(S,q,l)\right)_{\Sin} = \kb$, which gives strictly higher utility, i.e.,~$\wtild u(\kb,p) = \sum_{q\in \mathcal{Q}^l,S} (v^l(S,q) - p(S)) \wtild x(S,q,l) > u(\kb,p)$. 
Then, there exist at least two unit-demand agents $i,j \in \mathcal{\wtild Q}$, for which $ 1 = \wtild x(\wtild S,q,l)\neq x(\wtild S,q,l) = 0$, $(q,\wtild S) = (i,\wtild S_i), (j,\wtild S_j)$ and $ 0 = \wtild x(S,q,l)\neq x(S,q,l) = 1$, $(q,S) = (i,S_i), (j,S_j)$. Moreover, $v^l(\wtild S_i,i) - p(\wtild S_i) + v^l(\wtild S_j,j) - p(\wtild S_j) > v^l(S_i,i) - p(S_i) + v^l(S_j,j) - p(S_j) = b(i,l) + b(j,l)$. However, by \cref{equ:DSWLP1}, we must have $v^l(S,q) - p(S) \leq b(q,l)$ for all $\Sin$, a contradiction. Thus, no additional surplus can be generated by shifting allocations between winning unit-demand bidders. Moreover, no additional surplus can be generated by allocating to non-winning unit-demand agents, as, by the contraposition of \hyperref[fact:F4]{(F4)} we must have $p(S) \geq v^l(S,q)$ for all $\Sin$, for all $q\in \Qbarset\setminus \mathcal{\wtild Q}$. Thus, $\wtild u(\kb,p) \leq u(\kb,p)$.

(c): If a step in the incremental cost function $\DC(S,\cdot)$ is allocated, i.e.,~$y(S,r) > 0$, then the contrapositive of \hyperref[fact:F6]{(F6)} implies $\psi^G(\{p(S)\}) \geq \DC(S,r)$.
Together with \cref{prop:characteristic-dual} we have $p(S) \geq \left(d(\Succes(S),r) + \DC(\Succes(S),r)\right) + \DC(S,r)$, i.e.,~the seller always sells package $S$ at a weakly positive surplus.
Furthermore, if $p(S) > 0$, it follows by \hyperref[fact:F2]{(F2)} that $\sum_{S'\in\Predec_0(S)} x({S'},\Qbarset,\buyers) = y(S,\Rbarset) $, for all $\Sin$, i.e.,~the amount of all packages with cost connection to $S$ sold equals $y(S,\Rbarset)$ and no package assigned on some incremental cost function goes to waste.

Now we show that, if a positive surplus can be made on some incremental cost function step $r$ corresponding to package $S$, then it is assigned $y(S,r) = 1$. Suppose $p(S) > \DC(S,r) + d(\Succes(S),r) + \DC(\Succes(S),r)$, i.e.,~a strictly positive surplus is made on package $S$ (recall $d(S,r)\geq 0$ for all $S,r$). Then, we have
\begin{align*}
    d(S,r) & = 
    d(\Succes_0(S),r) + \DC(\Succes_0(S),r) - \left(d(\Succes(S),r) + \DC(\Succes(S),r)\right) - \DC(S,r) \\
    & \geq p(S) - \left(d(\Succes(S),r) + \DC(\Succes(S),r)\right) - \DC(S,r)\\
    & > 0
\end{align*}
By the contrapositive of \hyperref[csc:S3]{(S3)}, $d(S,r)>0$ implies that $y(S,r) = 1$.

We also show that, if a loss would be made by assigning step $r$ corresponding to package $S$, then $y(S,r) = 0$. Note that the loss is compared to not assigning the items in $S$ at all, or compared to assigning the items contained in $S$ as a different partition (with elements that are strict subsets of $S$). Thus, the shadow prices $d(S,r)$, which capture potential gains in these subsets, appear in the following equation. Let $p(S) < d(\Succes(S),r) + \DC(\Succes(S),r) + \DC(S,r)$, i.e.,~assigning package $S$ is not profitable at the given prices. For contradiction, suppose $y(S,r)>0$. Then, \cref{prop:characteristic-dual} applies, and
\begin{align*}
    \psi^G(\{p(S)\}) & = p(S) - \left(d(\Succes(S),r) + \DC(\Succes(S),r)\right)\\
    & < d(\Succes(S),r) + \DC(\Succes(S),r) + \DC(S,r) - \left(d(\Succes(S),r) + \DC(\Succes(S),r)\right)\\
    & = \DC(S,r)
\end{align*}
Then, by \hyperref[fact:F6]{(F6)} we must have $y(S,r) = 0$.

(d) We claim that, given a partition of supply that is a solution to SWLP, the seller cannot improve her utility by choosing a different partition of supply. To see this, recall \cref{lem:one-to-one-mapping}, which states that the mapping from the assignment on incremental cost functions $(y(S,r))_{\Sin,r \leq \Rbar}$ to a package multiset $(k_S)_{\Sin}$ is one-to-one. In (c), we have shown that the seller cannot improve her surplus given the assignment on incremental cost functions, given the dual prices. By \cref{lem:one-to-one-mapping}, the resulting partition of supply is unique and also optimal for the seller.
\end{proof2}\\

\begin{proof2}[Proof of \cref{prop:IP-LP-existence}]\phantomsection\label{proof:prop:IP-LP-existence}
    Let $O_{\text{SWP}}$ ($O_{\text{SWLP}}$, $O_{\text{DSWLP}}$) denote the value of an optimal solution to SWP (SWLP, DSWLP), respectively. First, suppose $O_{\text{SWP}} = O_{\text{SWLP}}$. Then, there exists $\{x(S,q,l),y(S,r)\}$ as an optimal solution to SWP and SWLP, and $\{x(S,q,l),y(S,r)\}$ is efficient. By \cref{prop:primal-dual-CE}, the dual variables $\{p(S)\}$ of DSWLP support this allocation as a package-linear pricing Walrasian equilibrium.

    Now suppose that there exists an equilibrium, i.e.,~prices $\{p(S)\}\geq 0$ that support $\{x(S,q,l)$, $y(S,r)\}$ as an equilibrium allocation. By \cref{prop:efficientPLW-multi-unit}, the allocation is efficient. Let $b(q,l) := v^l(S,q) - p(S)$ for all $q\in\Qbarset,l\in\buyers: x(S,q,l)>0$, and let $d(S,r):= \psi^G_S(\{p(S)\}) - \DC(S,r)$ for all $\Sin, r\leq \rtild_S:y(S,r)>0$. The dual variables $d(S,r)$, $\Sin,r\leq\rtild_S$, are defined recursively by \cref{prop:characteristic-dual}: for $j\in N$, we have $\psi_S^G(\{p(j)\}) = p(j)$, so, by \cref{prop:pricing}, $d(j,r):= p(j) - \DC(j,r)$ for all $j\in N, r \leq \Rbar$. Given $d(j,r)$, one can go on to define $d(S,r)$ for all $S\in\DPredec(j), j\in N$, etc. The $d(S,r)$ are the seller's surplus on each individual supply step, and the $b(q,l)$ are the surplus of each unit-demand agent. Because $\{p(S)\}$ are Walrasian equilibrium prices, it must be that the surpluses are positive, i.e.,~$d(S,r),b(q,l) \geq 0$. Together with their definitions, this implies that $b(q,l), d(S,r)$, and $p(S)$ are feasible in DSWLP.

    We now show that the seller's revenue $\sum_S p(S) Y_S$ is equivalent to the sum of the seller's revenue corresponding to each incremental cost function, $\sum_{S,r:y(S,r)=1} \psi^G_S(\{p(S)\})$. Using \cref{lem:combinatorics-lem}, we have
    \begin{align*}
        & \sum_{S,r:y(S,r)=1} \psi^G_S(\{p(S)\}) = \sum_{S,r:y(S,r)=1}  p(S) - \left(d(\Succes(S),r) + \DC(\Succes(S),r)\right)
    \end{align*}
    Moreover, recall that the quantity assigned on an incremental cost function $\DC(S,\cdot)$ equals the sum of all packages sold that have a cost connection to $S$, i.e., $\sum_{r} y(S,r) = \sum_{S'\in\Predec_0(S)}Y_{S'} =  Y_S + \sum_{S'\in\Predec(S)}Y_{S'}$. Thus, we further write
    \begin{align*}
        & \sum_{S,r:y(S,r)=1} \psi^G_S(\{p(S)\}) = \sum_{S} p(S) Y_S  + F, \quad \text{where}\\
        F := & \sum_{S} p(S) \left(\sum_{S'\in\Predec(S)} Y_{S'}\right)
        - \sum_{S}\left(d(\Succes(S),r) + \DC(\Succes(S),r)\right) \left(\sum_{S'\in\Predec_0(S)} Y_{S'}\right)
    \end{align*}
    Using \cref{prop:pricing}, we substitute for $p(S))$ and obtain
    \begingroup
        \allowdisplaybreaks
        \begin{align*}
        F = & \sum_{S} \left(d(\Succes_0(S),r) + \DC(\Succes_0(S),r)\right) \left(\sum_{S'\in\Predec(S)} Y_{S'}\right)\\
        - & \sum_{S}\left(d(\Succes(S),r) + \DC(\Succes(S),r)\right) \left(\sum_{S'\in\Predec_0(S)} Y_{S'}\right) \\
        = & \sum_{S} \left(d(S,r) + \DC(S,r)\right) \left(\sum_{S'\in\Predec(S)} Y_{S'}\right)\\
        - & \sum_{S}\left(\sum_{S'\in \Succes(S)} d(S',r) + \DC(S',r)\right) Y_S
    \end{align*}
    \endgroup
    The term $F$ is equal to zero because of symmetry: the sum of the products of each node with the sum of its predecessors is equal to the sum of the products of each node with its successors. Because we are taking the sum across all nodes, it does not matter in which direction we traverse the directed graph. Thus, we have shown that $\sum_S p(S) Y_S = \sum_{S,r:y(S,r)=1} \psi^G_S(\{p(S)\})$.
    
    Finally, efficiency of the package-linear Walrasian equilibrium (cf.~\cref{prop:efficientPLW-multi-unit}) implies that the allocation $\{x(S,q,l),y(S,r)\}$ is optimal in SWLP. By strong duality, $O_{\text{SWLP}} = O_{\text{DSWLP}}$ holds. Thus, we have
    \begingroup
    \allowdisplaybreaks
    \begin{align*} 
        O_{\text{SWLP}} & = O_{\text{DSWLP}} \\
               & \stackrel{(1)}{\leq} \sum_{q,l} b(q,l) + \sum_{S,r} d(S,r)\\
               & \stackrel{(2)}{=} \sum_{S,q,l:x(S,q,l)=1} (v^l(S,q) - p(S)) + \sum_{S,r:y(S,r)=1} (\psi^G_S(\{p(S)\}) - \DC(S,r))\\
               & \stackrel{(3)}{=} \sum_{q,l,S}v^l(S,q)x(S,q,l)-%
               \sum_{S,r}\DC(S,r)y(S,r)\\
               & \stackrel{(4)}{=} O_{SWP}
    \end{align*}
    \endgroup
    (1) follows from the objective function of DSWLP. (2) follows by definition of $b(q,l)$ and $d(S,r)$ above. (3) follows because $\sum_{S,q,l:x(S,q,l)}p(S) = \sum_S Y_S p(S)$. Finally, (4) follows from the definition of SWP. Overall, $O_{\text{SWLP}} \leq O_{SWP}$. It also holds that $O_{\text{SWLP}} \geq O_{SWP}$ because any solution of SWP is feasible in SWLP, and the claim follows.
\end{proof2}\\

\begin{lemma}\label{lem:combinatorics-lem}
    The sets $x\subseteq z\in 2^N$, and a number $q$ with $|x|\leq q\leq |z|$ are given. Let $R:=\{y\in 2^N|x\subseteq y \subseteq z, |y|=q\}$. Then, $|R|=\binom{|z|-|x|}{q-|x|}$.
\end{lemma}
\begin{proof}
    This is a standard combinatorics problem. First, note that $q-|x|$ items can be added to $x$ such that $y$ contains $q$ items. These items also need to be different from those contained in $x$, and they need to be contained in $z$. Hence, there are $|z|-|x|$ different items, of which $q-|x|$ many can be added to $x$. This is possible in $\binom{|z|-|x|}{q-|x|}$ different ways.
\end{proof}

\begin{proof2}[Proof of \cref{prop:charac-function-complete-cfg}]\phantomsection\label{proof:prop:charac-function-complete-cfg}
    Let $W\subseteq S \subseteq S' \in 2^N$, and let $y(S,S',t,W)$ denote the amount allocated on the incremental cost function $\DC(S,\cdot)$ due to the cost connection of the allocated bundle $S'\supset_t W$, $t$ levels above $W$. Let $dist(x,y):= ||x|-|y||$ for any $x,y \in 2^N$. We first establish a series of facts.
    
    \noindent {\it Fact (1).} Given any reference supply set $W \subseteq S$, we can write the amount allocated on incremental cost function $\DC(S,\cdot)$ as
    \begin{align*}
        \sum_q y(S,r) = \sum_{t=dist(S,W)}^{n-|W|}\sum_{S'\supset_t W} y(S,S',t,W).
    \end{align*}
    
    \noindent {\it Fact (2).} Given $S,S',W$, let $r:= dist(S,W)$ and $t:=dist(S',W)$. By \cref{lem:combinatorics-lem} above, there exist $\binom{t}{r}$ incremental cost functions $\DC(S,\cdot)$ relative to $\DC(W,\cdot)$, on which the amount $y(S,S',t,W)$ is allocated.
    
    \noindent {\it Fact (3).} $y(S,S',t,W)$ does not depend on $S$. If a step on the incremental cost function $\DC(S',\cdot)$ is allocated, then a step on each incremental cost function $S \subseteq S'$ is allocated. Thus, the allocation of bundles on incremental cost functions in the graph ``between'' $W$ and $S'$ \emph{due to the allocation of bundle $S'$} has to be the same amount.
    
    \noindent {\it Fact (4).} For any $t\geq 1$, we have, by the binomial theorem, $\sum_{z=0}^t \binom{t}{z}(-1)^z = 0$.
    
    \noindent Using the facts above (references in the equation below), we have
    {\everymath={\displaystyle}
    \begin{equation*}
        \begin{array}{cl}
        & \sum_{t=0}^{n-|S|}\sum_{r}\sum_{S' \supset_{t}S} (-1)^t y(S',r) \\
        = & \sum_{t=0}^{n-|S|}\sum_{S' \supset_{t}S} (-1)^t \sum_{r}y(S',r) \\
        \stackrel{(1)}{=} & \sum_{t=0}^{n-|S|}\sum_{S' \supset_{t}S} (-1)^t \sum_{t=dist(S',S)}^{n-|S|}\sum_{S''\supset_t S} y(S',S'',t,S) \\
        = & \sum_{t=0}^{n-|S|}\sum_{S' \supset_{t}S} (-1)^t \sum_{z=t}^{n-|S|}\sum_{S''\supset_z S} y(S',S'',z,S) \\
        \stackrel{(2),(3)}{=} & \sum_{t=0}^{n-|S|}(-1)^t \sum_{z=t}^{n-|S|} \binom{z}{t}\sum_{S''\supset_z S} y(\cdot,S'',z,S) \\
        \stackrel{(4)}{=} & \sum_{t=0}^{n-|S|} \sum_{S''\supset_t S} y(\cdot,S'',t,S) \sum_{z=0}^{t} (-1)^z\binom{t}{z}\\
        = & \sum_{S''\supset_0 S} y(\cdot,S'',0,S) \sum_{t=0}^{0} (-1)^t\binom{0}{t}\\
        = & y(\cdot,S,0,S)\\
        = & \phi^G_S(\{y(S,r)\}) 
        \end{array}
    \end{equation*}
    }
\end{proof2}

\subsection{Proofs for \cref{sec:equilibrium-existence}}\label{app:sec:proofs-equilibrium-existence}

\begin{proof2}[Proof of \cref{prop:ascending-auction}]\phantomsection\label{proof:prop:ascending-auction}
    The proof proceeds by analogy with \cite{Sun2014}, but the detailed arguments differ. The auction terminates at some time $t^*$, because buyers' values are finite, i.e.,~demand ceases at some point. The empty package is always priced at zero.
    
    Let $p^* = p(t^*)$ and let $S^{l*} = S^l(t^*)$. Furthermore, let $\partition^* = \partition(t^*)\in D^0(p^*)$ denote the supply set in $D^0(p^*)$ that is chosen at time $t^*$ by the seller. First, we establish an allocation $\pi^*$ such that $(p^*,\pi^*)$ constitutes a package-linear Walrasian equilibrium. Because at $p^*$ no package is overdemanded, for any buyer $l\in \buyers$ with $S^{l*} \neq\emptyset$, his demand $S^{l*}$ must be in $\partition^*$. If $\bigcup_{l\in\buyers}S^{l*} = N$ holds, let $\pi^*(l) = S^{l*}$ for all $l\in \buyers$ and $\pi^*(0) = \emptyset$. Then $(p^*, \pi^*)$ is a package-linear Walrasian equilibrium.
    If $\bigcup_{l\in\buyers}S^{l*} \subset N$, there is at least one package $B$ in the chosen supply set $\partition^*$ which is not demanded by any buyer at time $t^*$. By analogy with SY, $B$ is called a \emph{squeezed-out} package. We distinguish multiple cases:
   
    {\bfseries Case 1:} $p^*(B) = c^0(B)$.
    The final price of bundle $B$ is still fixed at the starting price, so $B$ was never overdemanded. If a buyer demanded it in some earlier round, this buyer demands now a different, more profitable package. Let $\partition_0^* = \{B\in\partition^*~|~p^*(B) = c^0(B) \text{ and } B\neq S^{l*} \text{ for all } l \in \buyers\}$ be the set of all squeezed-out packages. Let $\pi^*(0) = \bigcup_{B\in\partition^*_0} B$ and allocate $\pi^*(0)$ to the seller at zero cost. Let $\K_0^*$ denote the universe of all partitions of the items contained in $\partition_0^*$. Because $\partition^*\in D^0(p^*)$, we have
    \begin{align*}
        \sum_{B\in \gamma}\left[p^*(B) - c^0(B)\right] \leq
        \sum_{B\in\partition_0^*}\left[p^*(B) - c^0(B)\right] = 0
    \end{align*}
    for all $\gamma \in \K_0^*$. Hence, the seller is indifferent between selling $\pi^*(0)$ or not.
    
    {\bfseries Case 2:} $ p^*(B) > c^0(B)$.
    Package $B$ was demanded by some buyer in an earlier round. Denoting by $t$ the last round in which $B$ was demanded by some buyer $l$, $B$ may be allocated to buyer $l$ at the current price $p^*(B)$ by the auction rule. Thus, we need to demonstrate that it is still utility-maximizing for buyer $l$ to receive package $B$ at the current price. By the auction rule, we must have $\mathcal{V}^l(p(t)) = u^l(B,p(t)) = v^l(B) - p(t,B) \geq 1$ and $p^*(B) = p(t,B)$ or $p^*(B) = p(t,B) + 1$. Thus, we have for buyer $l$, who is allocated the squeezed-out package $B$,
    \begin{align}
        u^l(B,p^*) = v^l(B) - p^*(B) \geq 0 \label{equ:quasiLinB}
    \end{align}
    Now two sub-cases need to be distinguished:

    {\bfseries Case 2A:} If $S^{l*} = \emptyset$, assign buyer $l$ the squeezed-out bundle, i.e.,~$\pi^*(l) = B$. $S^{l*}\in D^l(p^*)$ and $S^{l*} = \emptyset$ imply that $\mathcal{V}^l(p^*) = 0 $. By definition of $\mathcal{V}^l$ we have $\mathcal{V}^l(p^*)\geq u^l(B,p^*)$. Together with \cref{equ:quasiLinB} this implies $u^l(B,p^*) = 0$, and hence $\pi^*(l)\in D^l(p^*)$.
    
    {\bfseries Case 2B:} If $S^{l*} \neq \emptyset$, assign buyer $l$ what he demanded at time $t^*$ and the squeezed-out bundle, i.e.,~$\pi^*(l) = S^{l*} \cup B$. Because the seller chose a supply set $\partition^* \ni \{S^{l*}, B\}$, we have
    \begin{align}
        p^*(S^{l*}) - c^0(S^{l*}) + p^*(B) - c^0(B) \geq p^*(\pi^*(l)) - c^0(\pi^*(l)). \label{equ:profitBundleSeller}
    \end{align}
    
    Superadditivity of $l$'s utility, subadditivity of the seller's cost, and $S^{l*}\in D^l(p^*)$ imply
     \begin{align}
        v^l(\pi^*(l)) & \geq v^l(S^{l*}) + v^l(B) \label{equ:superAdd} \\
        c^0(\pi^*(l)) & \leq c^0(S^{l*}) + c^0(B) \label{equ:subAdd} \\
        v^l(S^{l*}) - p^*(S^{l*}) & \geq v^l(\pi^*(l)) - p^*(\pi^*(l)). \label{equ:buyerLOpt}
    \end{align}
    From \cref{equ:profitBundleSeller,equ:subAdd} follows
    \begin{align} \label{equ:priceBundleSeller}
        p^*(S^{l*}) + p^*(B) \geq p^*(\pi^*(l)).
    \end{align}
    Then, using equation \cref{equ:priceBundleSeller}, \cref{equ:superAdd}, and \cref{equ:quasiLinB} (in this order), we obtain
    \begin{align*}
        v^l(\pi^*(l)) - p^*(\pi^*(l)) & \geq v^l(\pi^*(l)) - \left[p^*(S^{l*}) + p^*(B)\right] \\
        & \geq \left[v^l(S^{l*}) - p^*(S^{l*})\right] + \left[v^l(B) - p^*(B)\right]\\
        & \geq v^l(S^{l*}) - p^*(S^{l*}).
    \end{align*}
    By \cref{equ:buyerLOpt}, $v^l(\pi^*(l)) - p^*(\pi^*(l)) = v^l(\pi^*(l)) - \left[p^*(S^{l*}) + p^*(B)\right] = v^l(S^{l*}) - p^*(S^{l*})$, and thus
    \begin{align}
        p^*(\pi^*(l)) = p^*(S^{l*}) + p^*(B). \label{equ:priceBundleEqual}
    \end{align}
    Buyer $l$ is therefore happy to receive bundle $B$ in addition to his demanded bundle $A_k^*$, and pay the price that is set for the bundle $\pi^*(l)$. This process can be repeated for every squeezed-out bundle $B$ with $p^*(B) > c^0(B)$. Every buyer $l$ who is not allocated any squeezed-out bundle receives his demanded package, i.e.,~$\pi^*(l) = S^{l*}$. $\partition^*$ is a partition of $N$ chosen by the seller, and thus $(\pi^*(0), \dots, \pi^*(L))$ is an allocation of $N$. By \cref{equ:priceBundleEqual}, the seller's utility is
    \begin{align*}
        \sum_{l\in\buyers} \left[p^*(\pi^*(l)) - c^0(\pi^*(l))\right] = \sum_{A\in\partition^*} \left[p^*(A) - c^0(A)\right] = \mathcal{V}^0(p^*) 
    \end{align*}
    It follows that $(p^*,\pi^*)$ is a package-linear pricing Walrasian equilibrium.
\end{proof2}\\

\begin{proof2}[Proof of \cref{prop:seller-equivalence}]\phantomsection\label{proof:prop:seller-equivalence}
    In SY, the seller's supply correspondence is defined as
    \begin{align*}
        S(p) = \argmax_{\partition \in \K}\left\{\sum_{A\in \partition}p(A)\right\}
    \end{align*}
    In our ascending auction the seller's supply correspondence is defined as
    \begin{align*}
        D^0(p) = \argmax_{\partition \in \K} \left\{\sum_{A \in \partition}\left(p(A) -c^0(A)\right)\right \}
    \end{align*}
    In SY's ascending auction, in holds that $p(B) = v^0(B)$ for any bundle $B$ that is assigned to the seller during the procedure and $p(\pi(0)) = v^0(\pi(0))$. Hence, we have
    \begin{align*}
        S(p) & =\argmax_{\partition \in \K}\left\{\sum_{A\in \partition\setminus B}p(A) + v^0(B)\right\}\\
        & = \argmax_{\partition \in \K}\left\{\sum_{A\in\partition\setminus B} p(A)  + v^0(B) - v^0(N) \right\}\\
        & = \argmax_{\partition \in \K}\left\{\sum_{A\in\partition\setminus B} p(A)  - c^*(N\setminus B) \right\}\\
        & = \argmax_{\partition \in \K}\left\{\sum_{A\in\partition\setminus B} p(A)  - c^*\left(\bigcup_{A\in\partition\setminus B} A\right) \right\}
    \end{align*}
    $c^*$ is by definition the dual of $v^0$, and $c^*\left(\bigcup_{A\in\partition\setminus B} A\right)$ may be interpreted as the seller's cost function. Thus, part (i) and (ii) of the proposition follow.
\end{proof2}\\

\begin{proof2}[Proof of \cref{prop:setcover-submodular}]\phantomsection\label{proof:prop:setcover-submodular}
    To simplify notation, we write $c^*(v^0,S)$ as $c^*(S)$ for any $\Sin$. Let $S_1^{\mathrm{c}}, S_2^{\mathrm{c}} \in 2^N$ and $S_1^c \cap S_2^c = \emptyset$. Note that $c^*(N) = v^0(N)$ and $S_1^c \cap S_2^c = \emptyset \Leftrightarrow S_1 \cup S_2 = N$.  Because $v^0$ is superadditive we have
    \begin{align*}
        & v^0(S_1^{\mathrm{c}} \cup S_2^{\mathrm{c}}) \geq v^0(S_1^{\mathrm{c}}) + v^0(S_2^{\mathrm{c}}) \\
        \Leftrightarrow \quad & v^0(N) - c^*((S_1^{\mathrm{c}} \cup S_2^{\mathrm{c}})^{\mathrm{c}}) \geq 2 v^0(N) - c^*(S_1) - c^*(S_2) \\
        \Leftrightarrow \quad & c^*(S_1) + c^*(S_2) \geq c^*(S_1\cup S_2) + c^*(S_1\cap S_2)
    \end{align*}
    The proof for subadditive $v^0$ is analogous.
\end{proof2}\\

\begin{proof2}[Proof of \cref{lem:same-allocation}]\phantomsection\label{proof:lem:same-allocation}
    With a revenue-maximizing seller, an allocation $\pi$ is efficient if it holds for every allocation $\pi'$ that
    \begin{align} \label{equ:SYefficient}
        \sum_{l\in\agents}\left[v^l(\pi(l))\right] \geq \sum_{l\in\agents}\left[v^l(\pi'(l))\right]
    \end{align}
    With a utility-maximizing seller, an allocation $\pi$ is efficient if for every allocation $\pi'$ it holds that
    \begin{align} \label{equ:efficient}
        \sum_{l\in\buyers}\left[v^l(\pi(l)) - c^0(\pi(l))\right] \geq \sum_{l\in\buyers}\left[v^l(\pi'(l)) - c^0(\pi'(l))\right]
    \end{align}
    SY's ascending auction terminates in an efficient allocation $\pi$ in the sense of \cref{equ:SYefficient}. The extended ascending auction terminates in an efficient allocation $\pi$ in the sense of \cref{equ:efficient}. If the auctioneer's marginal costs are zero, the efficient allocation in the extended ascending auction is equivalent to the efficient allocation in the sense of \cref{equ:SYefficient}: running the auction with the set of buyers $\buyers' = \buyers + \{0\} = \agents$, \cref{equ:SYefficient} and \cref{equ:efficient} are equivalent and the claim follows.
\end{proof2}\\

\begin{proof2}[Proof of \cref{prop:extended-ascending-auction}]\phantomsection\label{proof:prop:extended-ascending-auction}
    Note that every conventional buyer bids identically in the ascending auction and the extended ascending auction, up to ties. We split the revenue-maximizing seller into $2^n$ dummy buyers, denoted $l_S, \Sin$. Define dummy $l_S$'s utility function as follows:
    \begin{align*}
        v^{l_S}(B):=\left\{ \begin{array}{ll}
            v^0(S) & \text{ if } B \supseteq S \\
            0 & \text{ otherwise }
        \end{array} \right.
    \end{align*}
    Each dummy $l_S$ has the highest bid on bundle $S$ among all dummies because $v^0$ is superadditive. Let dummy $l_S$ demand bundle $S$ whenever he weakly prefers $S$ to any other bundle except the empty set and let him demand the empty set when she weakly prefers to do so.
    
    Let the extended ascending auction start at $t=-1$ with starting prices $p(-1,S) = v^0(S)-1 ~\forall \Sin$. Let two instances of each dummy $l_S$ participate. Dummies $l_S, \Sin$ each demand bundle $S$. The auctioneer offers some supply set. Regardless of the non-dummy buyers' demand, each bundle $\Sin$ is overdemanded in $t=-1$, so prices in $t=0$ are increased by one. The dummies all demand the empty set for all $t=0,1,\dots$, so if at some round $t\geq 0$ the auction ends with squeezed-out bundles, they can be assigned to the dummies if they were the last to demand them. It is without loss of generality to stipulate that the squeezed-out bundles are allocated to dummies in this case, and not to regular buyers who might have demanded them at $t=-1$ as well.
    Then, in all rounds $t=0,1,\dots$, the supply correspondence and the demand correspondences are chosen to maximize identical utility functions in both auctions. Hence, the supply and demand correspondences are identical in every round of both auctions, and it follows that an identical price path resulting in the same allocation exists.
\end{proof2}

\section{Additional examples}

\subsection{Packaging costs between identical items}\phantomsection\label{sec:identical-items}

If the seller has partition preferences over identical items (or if there are complementarities between identical items on the buyers' side), one can appropriately relabel items and adjust valuations and costs.\footnote{From a market design perspective, it is most efficient in terms of computational complexity to only relabel those items on which partition preferences or complementarities are expected.} We illustrate this with an example.
\begin{example}\label{ex:preprocessing}
    There are two items $A$ and $B$ supplied with $\Omega_A = \Omega_B = 2$. We wish to allow the package $\{AA\}$ to have its own price $p(AA)$ not necessarily equal to $2p(A)$.
    The values are given by $v(A,q,l),v(B,q,l),v(AA,q,l),v(AB,q,l),v(AAB,q,l)$, $l = 1,2$, $q=1,2$. We give each unit of $A$ its own index, i.e.,~$N:=\{{A_1}, {A_2}, B\}$, and obtain the values $v(A_1, q, l)$,$v(A_2,q,l)$, $v(B,q,l), v(A_1,A_2,q,l)$, $v(A_1B_,q,l), v(A_2B,q,l)$, and $v(A_1A_2B, q, l)$, where $v(A_1,q,l) = v(A_2,q,l)$, $v(A_1A_2,q,l) = v(AA,q,l)$, $v(A_1B,q,l) = v(A_2B,q,l) = v(AB,q,l)$ and $v(A_1A_2B,q,l) = v(AAB,q,l)$.

    The seller submits incremental cost functions $\DC(A,\cdot)$, $\DC(B$, $\DC(AA,\cdot)$, $\DC(AB,\cdot)$, and $\DC(AAB,\cdot)$, and a cost function graph defining the cost connections between these packages as shown in \cref{fig:cfg-identical-items}. The transformed incremental cost functions $\DC(A_1,\cdot)$, $\DC(A_2,\cdot)$, $\DC(A_1 A_2,\cdot)$, $\DC(A_1 B,\cdot)$, $\DC(A_2 B,\cdot)$, $\DC(A_1 A_2 B,\cdot)$ are such that $\DC(A_1,1) = \DC(A,1)$, $\DC(A_2,1) = \DC(A,2)$, $\DC(A_1 B,1) = \DC(AB,1)$, $\DC(A_2 B,1) = \DC(AB,2)$, $\DC(A_1 A_2,1) = \DC(AA,1)$, and $\DC(A_1 A_2 B,1) = \DC(AAB,1)$. $\DC(B,\cdot)$ remains unchanged and all other $\DC(S,r)$ are set to $\infty$. The cost function graph is adjusted as shown in \cref{fig:augmented-cfg}, where outgoing edges are implied by the original graph in \cref{fig:cfg-identical-items}.
    \begin{center}
    \renewcommand{\tabularxcolumn}[1]{b{#1}}
    \begin{tabularx}{0.9\textwidth}{>{\centering\arraybackslash}X >{\centering\arraybackslash}X}
        \centering
        \includegraphics[scale=1]{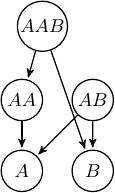}
        \captionof{figure}{CFG with package $AA$}
        \label{fig:cfg-identical-items}
        &
        \includegraphics[scale=1]{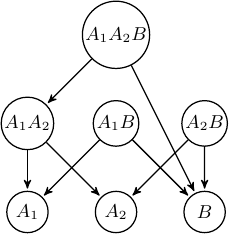}
        \captionof{figure}{Augmented CFG}
        \label{fig:augmented-cfg}
    \end{tabularx}
    \end{center}
    Note that the auction prices must satisfy $p(A) = \min\{p(A_1), p(A_2)\}$, $p(AA) = p(A_1A_2)$, and $p(AB) = \min\{p(A_1B),p(A_2B)\}$. If $A_1$ and $A_2$ are both allocated to buyers, it must hold that $p(A_1) = p(A_2)$, and similarly for $p(AB)$, because of $v(A_1,q,l) = v(A_2,q,l) = v(A,q,l)$ and the constraints of DSWLP.
\end{example}

\subsection{Ascending auctions}

\begin{example}\label{ex:ref_vs_utilmax}
    Buyers are labeled L1 to L6. Their values and the seller's values $v^0$ and dual marginal costs $c^0$ are given in \cref{tab:ex:rev_vs_utilmax}. \cref{tab:ex:rev_vs_utilmax2} details the ascending auction where the seller maximizes revenue or, equivalently, utility based on the cost function $C^0(\partition) = c^0\left(\bigcup_{S\in\partition}S\right)$, in each round. \cref{tab:ex:rev_vs_utilmax3} details the extended ascending auction where the seller maximizes utility based on the cost function $\wtild{C}^0(\partition) = \sum_{S\in\partition} c^0\left(S\right)$ in each round. If the seller maximizes revenue, the two individual items $A$ and $B$ are allocated, e.g.,~to L1 and L3. If the seller maximizes utility based on $\wtild{C}^0$, bundle $AB$ is allocated, e.g.,~to L5.
    \begin{table}[ht]
    \hspace{0.5cm}
    \begin{minipage}{.23\textwidth}
        \centering
        \begin{adjustbox}{max width=0.8\textwidth,center}
        \begin{tabular}{cccc}
        \toprule
                 & $A$ & $B$ & $AB$ \\
        \midrule
        $v^1$    & 5     & 0     & 5 \\
        $v^2$    & 5     & 0     & 5 \\
        $v^3$    & 0     & 7     & 7 \\
        $v^4$    & 0     & 7     & 7 \\
        $v^5$    & 0     & 0     & 11 \\
        $v^6$    & 0     & 0     & 11 \\
        \midrule
        $v^0$ & 2     & 4     & 8 \\
        $c^0$ & 4     & 6     & 8 \\
        \bottomrule
        \end{tabular}%
        \end{adjustbox}
        \captionof{table}{Values and costs}
        \label{tab:ex:rev_vs_utilmax}%
    \end{minipage}
    \hspace{0.2cm}
    \begin{minipage}{.7\textwidth}
        \centering
        \begin{adjustbox}{max width=0.9\textwidth,center}
        \begin{tabular}{llccccccc}
            \toprule
            \multicolumn{1}{c}{\multirow{2}[2]{*}{current price}} & \multicolumn{1}{c}{\multirow{2}[2]{*}{supply set}} & \multicolumn{6}{c}{demand}                    & \multirow{2}[2]{*}{squeezed-out} \\
                  &       & L1    & L2    & L3    & L4    & L5    & L6    &  \\
            \midrule
            $p(0) = (2,4,8)$ & $\{AB\}$ & $A$   & $A$   & $B$   & $B$   & $AB$  & $AB$  &  \\
            $p(1) = (3,5,9)$ & $\{AB\}$ & $A$   & $A$   & $B$   & $B$   & $AB$  & $AB$  &  \\
            $p(2) = (4,6,10)$ & $\{AB\}$ & $A$   & $A$   & $B$   & $B$   & $AB$  & $AB$  &  \\
            $p(3) = (5,7,11)$ & $\{A,B\}$ & $\emptyset$ & $\emptyset$ & $\emptyset$ & $\emptyset$ & $\emptyset$ & $\emptyset$ & A,B \\
            \bottomrule
        \end{tabular}%
        \end{adjustbox}
        \caption{Ascending auction}
        \label{tab:ex:rev_vs_utilmax2}%
        \vspace{25pt}
        \centering
        \begin{adjustbox}{max width=0.9\textwidth,center}
        \begin{tabular}{llccccccc}
            \toprule
            \multicolumn{1}{c}{\multirow{2}[2]{*}{current price}} & \multicolumn{1}{c}{\multirow{2}[2]{*}{supply set}} & \multicolumn{6}{c}{demand}                    & \multirow{2}[2]{*}{squeezed-out} \\
                  &       & L1    & L2    & L3    & L4    & L5    & L6    &  \\
            \midrule
            $p(0) = (4,6,8)$ & $\{AB\}$ & $A$   & $A$   & $B$   & $B$   & $AB$  & $AB$  &  \\
            $p(1) = (5,7,9)$ & $\{A,B\}$ & $\emptyset$ & $\emptyset$ & $\emptyset$ & $\emptyset$ & $AB$  & $AB$  &  \\
            $p(2) = (5,7,10)$ & $\{AB\}$ & $\emptyset$ & $\emptyset$ & $\emptyset$ & $\emptyset$ & $AB$  & $AB$  &  \\
            $p(3) = (5,7,11)$ & $\{AB\}$ & $\emptyset$ & $\emptyset$ & $\emptyset$ & $\emptyset$ & $\emptyset$ & $\emptyset$ & AB \\
            \bottomrule
        \end{tabular}%
        \end{adjustbox}
        \caption{Modified ascending auction}
        \label{tab:ex:rev_vs_utilmax3}%
    \end{minipage}
    \end{table}%
\end{example}

\end{document}

%% file: preamble.tex
\usepackage{amssymb}
\usepackage{amsfonts}
\usepackage{amsmath}
\usepackage{amsthm}
\usepackage[nohead]{geometry}
\usepackage[singlespacing]{setspace}
\usepackage[bottom,symbol]{footmisc}
\usepackage{indentfirst}
\usepackage{endnotes}
\usepackage{graphicx}
\usepackage{rotating}
\usepackage[utf8]{inputenc}
\usepackage[english]{babel}
\usepackage[authoryear]{natbib}
\usepackage{enumitem}
\usepackage{booktabs}
\usepackage{array}
\usepackage[dvipsnames]{xcolor}
\definecolor{mediumblue}{rgb}{0.0, 0.0, 0.95}
\usepackage{hyperref}
\hypersetup{
colorlinks=true,
linkcolor=mediumblue,
urlcolor=mediumblue,
citecolor=black,
pdftitle={Selling Multiple Complements with Packaging Costs},
}
\usepackage{url}
\usepackage{adjustbox}
\usepackage{mathrsfs}
\usepackage{bigstrut}
\usepackage{dsfont}
\usepackage{bbm}
\usepackage{mathtools}
\usepackage{physics}
\usepackage{breqn}
\usepackage{caption}
\usepackage{graphicx}
\usepackage{subcaption}
\usepackage{multirow}
\usepackage{longtable}
\usepackage{subfloat}
\usepackage{scalerel}
\usepackage{leftidx}
\usepackage{thmtools}
\usepackage{tabularx}
\usepackage{tabularray}
\usepackage{bm}
\usepackage[capitalise]{cleveref}
\usepackage{standalone}
\usepackage{bbding}
\usepackage[ruled,vlined]{algorithm2e}
\usepackage{empheq, cases}
\usepackage{authblk}
\newcommand\blfootnote[1]{
  \begingroup
  \renewcommand\thefootnote{}\footnote{#1}
  \addtocounter{footnote}{-1}
  \endgroup
}

\SetAlFnt{\small}
\SetAlCapFnt{\small}
\SetAlCapNameFnt{\small}
\SetAlCapHSkip{0pt}
\IncMargin{-\parindent}
\theoremstyle{definition}
\newtheorem{theorem}{Theorem}
\newtheorem{corollary}{Corollary}
\newtheorem{definition}{Definition}
\newtheorem{example}{Example}
\newtheorem{assumption}{Assumption}

\newtheorem{observation}{Observation}

\theoremstyle{plain}
\newtheorem{lemma}{Lemma}
\newtheorem{proposition}{Proposition}

\newenvironment{proof2}[1][Proof]{\noindent\textit{#1.} }{\ \hfill$\square$ }

\renewcommand\thmcontinues[1]{Continued}
\renewcommand*{\thefootnote}{\fnsymbol{footnote}}

\newcommand{\constraint}[1]{\refstepcounter{equation}\label{#1}(\theequation)}

\newcommand{\csclabel}[2]{\label{#1}(#2)}

\renewcommand{\vec}[1]{\bm{#1}}
\newcommand{{\pbold}}{\vec{p}}
\newcommand{\kb}{\vec{k}}
\newcommand{{\sbold}}{\vec{s}}

\newcommand{\xb}{\vec{x}}

\newcommand{\K}{\mathcal{K}}
\newcommand{\Z}{\mathbb{Z}}
\newcommand{\N}{\mathbb{N}}
\newcommand{\R}{\mathbb{R}}

\newcommand{\buyers}{[L]}

\newcommand{\msets}{\Z_+^{2^n}}
\newcommand{\agents}{\buyers_0}
\newcommand{\Qbar}{\widebar{Q}}
\newcommand{\Rbar}{\widebar{R}}
\newcommand{\Qbarset}{[\Qbar]}
\newcommand{\Rbarset}{[\Rbar]}
\newcommand{\DC}{{\Delta c}}
\newcommand{\rtild}{\widetilde{r}}
\newcommand{\Sin}{{S\in2^N}}
\newcommand{\rt}{{:\,}}
\newcommand{\Succes}{R^+}
\newcommand{\Predec}{R^-}
\newcommand{\DSucces}{N^+}
\newcommand{\DPredec}{N^-}
\newcommand{\partition}{\kb}

\newcommand{\mstar}{\scalebox{1.2}{$\star$}}

\DeclareMathOperator*{\argmax}{arg\,max}
\newcommand{\msetsum}{\sum}
\newcommand{\wtild}{\widetilde}

\makeatletter
\let\save@mathaccent\mathaccent
\newcommand*\if@single[3]{%
  \setbox0\hbox{${\mathaccent"0362{#1}}^H$}%
  \setbox2\hbox{${\mathaccent"0362{\kern0pt#1}}^H$}%
  \ifdim\ht0=\ht2 #3\else #2\fi
  }
\newcommand*\rel@kern[1]{\kern#1\dimexpr\macc@kerna}
\newcommand*\widebar[1]{\@ifnextchar^{{\wide@bar{#1}{0}}}{\wide@bar{#1}{1}}}
\newcommand*\wide@bar[2]{\if@single{#1}{\wide@bar@{#1}{#2}{1}}{\wide@bar@{#1}{#2}{2}}}
\newcommand*\wide@bar@[3]{%
  \begingroup
  \def\mathaccent##1##2{%
    \let\mathaccent\save@mathaccent
    \if#32 \let\macc@nucleus\first@char \fi
    \setbox\z@\hbox{$\macc@style{\macc@nucleus}_{}$}%
    \setbox\tw@\hbox{$\macc@style{\macc@nucleus}{}_{}$}%
    \dimen@\wd\tw@
    \advance\dimen@-\wd\z@
    \divide\dimen@ 3
    \@tempdima\wd\tw@
    \advance\@tempdima-\scriptspace
    \divide\@tempdima 10
    \advance\dimen@-\@tempdima
    \ifdim\dimen@>\z@ \dimen@0pt\fi
    \rel@kern{0.6}\kern-\dimen@
    \if#31
      \overline{\rel@kern{-0.6}\kern\dimen@\macc@nucleus\rel@kern{0.4}\kern\dimen@}%
      \advance\dimen@0.4\dimexpr\macc@kerna
      \let\final@kern#2%
      \ifdim\dimen@<\z@ \let\final@kern1\fi
      \if\final@kern1 \kern-\dimen@\fi
    \else
      \overline{\rel@kern{-0.6}\kern\dimen@#1}%
    \fi
  }%
  \macc@depth\@ne
  \let\math@bgroup\@empty \let\math@egroup\macc@set@skewchar
  \mathsurround\z@ \frozen@everymath{\mathgroup\macc@group\relax}%
  \macc@set@skewchar\relax
  \let\mathaccentV\macc@nested@a
  \if#31
    \macc@nested@a\relax111{#1}%
  \else
    \def\gobble@till@marker##1\endmarker{}%
    \futurelet\first@char\gobble@till@marker#1\endmarker
    \ifcat\noexpand\first@char A\else
      \def\first@char{}%
    \fi
    \macc@nested@a\relax111{\first@char}%
  \fi
  \endgroup
}
\makeatother